\documentclass{amsart}

\usepackage{graphicx}

\usepackage{proof}
\usepackage{stmaryrd}
\usepackage{lscape}
\usepackage{afterpage}
\usepackage{rotating}
\usepackage{graphicx}
\usepackage[all]{xy}
\usepackage{stmaryrd}
\usepackage{amssymb}
\usepackage{amsmath}
\usepackage{mathabx}
\usepackage{url}
\usepackage{hyperref}
\usepackage[normalem]{ulem}

\usepackage{xcolor}
\definecolor{darkgreen}{RGB}{20,200,10}
\newtheorem{theorem}{Theorem}
\newtheorem{proposition}{Proposition}
\newtheorem{corollary}{Corollary}
\newtheorem{lemma}{Lemma}
\newtheorem{definition}{Definition}

\newtheorem{nexample}{Example} %{\bfseries}{\rmfamily}
\newtheorem{nclaim}{Claim} %{\bfseries}{\rmfamily}

\newcommand\defn[1]{{\bf #1}}
  
\newcommand\eqdf{\mathbin{=_{\mathrm{df}}}}

\newcommand\I{\mathsf{I}}
\newcommand\ot{\otimes}
\newcommand\C{\mathcal{C}}
\newcommand\M{\mathcal{M}}
\newcommand\al{\alpha}
\newcommand\lam{\lambda}
\newcommand\n{\mathop{-}}
\newcommand\Var{\mathsf{At}}

\newcommand\tto[1][]{\overset{#1}\Rightarrow}
\newcommand\Tm{\mathsf{Fma}}
\newcommand\id{\mathsf{id}}
\newcommand\comp{\circ}
\newcommand\dcomp{\mathsf{comp}}

\newcommand\ax{\mathsf{id}}
\newcommand\uf{\mathsf{shift}}
\newcommand\IL{\I\mathsf{L}}
\newcommand\otL{\ot\mathsf{L}}
\newcommand\IR{\I\mathsf{R}}
\newcommand\otR{\ot\mathsf{R}}

\newcommand\axatm{\ax^{\mathrm{foc}}}
\newcommand\otRfoc{\otR^{\mathrm{foc}}}
\newcommand\IRfoc{\IR^{\mathrm{foc}}}

\newcommand\ILinv{\IL^{-1}}
\newcommand\otLinv{\otL^{-1}}

\newcommand\otLstar{\mathsf{L}}
\newcommand\otLinvstar{\mathsf{L}^{-1}}

\newcommand\scut{\mathsf{scut}}
\newcommand\ccut{\mathsf{ccut}}
\newcommand\scutR{\mathsf{scut}_{\mathsf{R}}}
\newcommand\ccutR{\mathsf{ccut}_{\mathsf{R}}}
\newcommand\cut{\mathsf{cut}}

\newcommand\sound{\mathsf{sound}}
\newcommand\cmplt{\mathsf{cmplt}}
\newcommand\strcmplt{\mathsf{strcmplt}}

\newcommand\switch{\mathsf{switch}}
\newcommand\focus{\mathsf{focus}}
\newcommand\emb{\mathsf{emb}}
\newcommand\embL{\mathsf{emb}_{\mathsf{L}}}

\newcommand\focderivs{\mathsf{focderivs}}
\newcommand\fskmaps{\mathsf{fskmaps}}

\newcommand\proofbox[1]{\begin{tabular}{c} #1 \end{tabular}}

\newcommand\lb{\langle\hspace*{-0.9mm}\langle}

\newcommand\asem[1]{\llbracket #1 \rrbracket}
\newcommand\csem[2]{#1\, \lb #2 \rrbracket}
\newcommand\ssem[1]{\llbracket #1 \lb\, }

\newcommand\dotminus{\dotdiv} % \stackrel{.}{-}}
\newcommand\Ptd{\mathbf{Ptd}}

\newcommand\inl{\mathsf{inl}}
\newcommand\inr{\mathsf{inr}}

\newcommand\e{\mathsf{e}}
\newcommand\m{\mathsf{m}}

\newcommand\JJ{\mathcal{J}}

\newcommand\Lan{\mathrm{Lan}}

\newcommand\Cat{\mathbf{Cat}}
\newcommand\Fsk{\mathbf{Fsk}}

\newcommand\mseq[2]{#1 \longrightarrow #2}
\newcommand\stseq[4][]{#2 \mid #3 \overset{#1}{\longrightarrow} #4}
\newcommand\stseqL[4][]{#2 \mid #3 \overset{#1}{\longrightarrow_{\mathsf{L}}} #4}
\newcommand\stseqR[4][]{#2 \mid #3 \overset{#1}{\longrightarrow_{\mathsf{R}}} #4}
\newcommand\stseqLR[4]{#1 \mid #2 \longrightarrow_{#4} #3}

\newcommand\multi[3][]{#2 \overset{#1}\longrightarrow #3}
\newcommand\tight[4][]{#2\mid#3 \overset{#1}\longrightarrow #4}
\newcommand\loose[3][]{\n\mid#2 \overset{#1}\longrightarrow #3}
\let\flexi=\tight

\renewcommand\arraystretch{3}

\changenotsign

\begin{document}
\title[The Sequent Calculus of Skew Monoidal Categories]{The Sequent Calculus of \\ Skew Monoidal Categories}
\thanks{This article is a revised and extended version of a paper presented at MFPS 2018 \cite{UVZ:seqsmcMFPS18}.}
\author{Tarmo Uustalu}
\address{Reykjavik University, Iceland, and Tallinn University of Technology, Estonia}
\email{tarmo@ru.is}
  
\author{Niccol\`o Veltri}
\address{Tallinn University of Technology, Estonia}
\email{niccolo@cs.ioc.ee}

\author{Noam Zeilberger}
\address{\'Ecole Polytechnique, Palaiseau, France}
\email{noam.zeilberger@lix.polytechnique.edu}

%\abstract*{\input abstract}
\begin{abstract}
  Szlach\'anyi's skew monoidal categories are a well-motivated
  variation of monoidal categories in which the unitors and associator
  are not required to be natural isomorphisms, but merely natural
  transformations in a particular direction. We present a sequent
  calculus for skew monoidal categories, building on the recent
  formulation by one of the authors of a sequent calculus for the
  Tamari order (skew semigroup categories).  In this calculus,
  antecedents consist of a \emph{stoup} (an optional formula)
  followed by a context, % (a list of formulae),
  and the connectives %unit and tensor
  behave like in the standard
  monoidal sequent calculus except that the left rules may only be
  applied in stoup position. We prove that this calculus is sound and
  complete with respect to existence of maps in the free skew monoidal
  category, and moreover that it captures equality of maps once a
  suitable equivalence relation is imposed on derivations.  We then
  identify a subsystem of focused derivations and establish that it
  contains exactly one canonical representative from each equivalence
  class. This \emph{coherence theorem} leads directly to simple procedures
  for deciding equality of maps in the free skew monoidal category and for
  enumerating any homset without duplicates.
  %(in particular, for deciding existence of a map).
  Finally, and in the spirit of
  Lambek's work, we describe the close connection between this
  proof-theoretic analysis and Bourke and Lack's recent
  characterization of skew monoidal categories as left representable
  skew multicategories.
  We have formalized this development in the dependently typed programming language Agda.

\end{abstract}

\maketitle
  
%% \keywords{skew monoidal categories \and substructural logics \and sequent calculus
%%   \and nonstandard sequent forms \and cut admissibility \and focusing
%%   \and Agda}

%% \tableofcontents
%% \nz{Table of contents included for ease of navigation while editing. Remove for final version.}
  
\section{Introduction}

Skew monoidal categories of Szlach\'anyi \cite{Szl:skemcb} are a
variation of monoidal categories \cite{ML:natac,Benabou}
in which the unitors and associator are not required to be natural isomorphisms but only natural transformations
in a particular direction.  Szlach\'anyi's original motivation for naming the concept was
the observation that bialgebroids provide natural examples of skew monoidal categories.
In a different context, the first author of
this paper ran into skew monoidal categories studying a generalization
of monads to functors between different categories: relative monads
\cite{ACU:monnnb}.

Szachl\'anyi's paper was immediately noticed by Street, Lack and colleagues who have by now published a whole series of followup works
\cite{LS:skemsw,LS:triosm,BGLS:catss,BL:free,BL:multi}.
Although the definition of skew monoidal category is simple, it entails some remarkably subtle properties.
For example, while Mac~Lane's coherence theorem for monoidal categories is often
summarized as ``all diagrams commute'', this is no longer true in the skew monoidal case: it is possible to have more than one map between a pair of objects in the free skew monoidal category on a set of generators (even for one generator).
Also, it is not so easy to give a simple necessary and sufficient condition for the existence of
such a map. Curiously, there \emph{is} at most one map between any two
objects in the free skew \emph{semigroup} category: indeed, this generates a natural partial order on fully bracketed words, known as the \emph{Tamari order} \cite{TamariFestschrift,Tamari1951phd}.
Multiple maps therefore somehow originate from the presence of the unit.

As a step towards the coherence
problem and taking a rewriting approach, Uustalu \cite{Uus:cohsmc}
showed that there is at most one map between an object and an object
in a certain normal form, and exactly one map between an object and that
object's normal form. In another direction, Lack and Street \cite{LS:triosm} addressed
the problem of determining equality of maps by proving that there is a faithful,
structure-preserving functor $\Fsk \to \Delta_\bot$ from 
the free skew monoidal category on one generating object to the category of finite non-empty ordinals and
first-element-and-order-preserving functions (which is a strictly
associative skew-monoidal category under the ordinal sum with skew unit 1).  This approach was
further elaborated by Bourke and Lack \cite{BL:free} with a more
explicit description of the morphisms of $\Fsk$, both papers
taking for granted an older analysis of the Tamari order by Huang and Tamari (1972).

In this paper we introduce a sequent calculus formulation of skew monoidal categories, building on a recent proof-theoretic analysis of the
Tamari order by the third author \cite{Zei}.
He observed that the Tamari order is precisely captured by a sequent calculus very
similar to Lambek's original ``syntactic calculus'' \cite{Lam:matss} (what is nowadays referred to as the Lambek calculus, or as a fragment of non-commutative intuitionistic linear logic),
but with the following restrictions: tensor is the only logical connective, and the tensor
left rule is restricted to only apply to the leftmost formula in the antecedent. 
This calculus admits a strong form of cut-elimination known as \emph{focusing} (the terminology comes from linear logic \cite{Andreoli92}), which yields the coherence theorem that valid entailments of the Tamari order are in one-to-one correspondence with focused derivations.

As we will see, the situation becomes more subtle with the addition of a unit.
Sequents now need to have an explicit ``stoup'' (again, the terminology is from linear logic \cite{Girard1991LC}) corresponding to a distinguished position at the left end of an antecedent that can either be empty or contain a formula.
The left rules are still restricted to only apply at the leftmost end of an antecedent to the formula within the stoup, but now there is also an explicit structural rule for moving a formula from the context to the stoup (on the bottom-up 
reading of the rule).
We will see that this has interesting consequences for the metatheory of the sequent calculus, and the end result will be a new coherence theorem for skew monoidal categories with two practical applications: 1.~a simple algorithm for deciding equality of maps in $\Fsk$, and 2.~a simple algorithm for duplicate-free enumeration of all of the maps between any pair of objects in $\Fsk$.

%% We will develop the metatheory of this
%% sequent calculus, and see that the presence of the stoup is crucial
%% for adequacy with respect to skew monoidal categories.

%% After establishing the admissibility of two cut rules (stoup cut and
%% context cut), we prove that the sequent calculus is sound and complete
%% with respect to the free skew monoidal category in the sense that
%% morphisms can be mapped to derivations and vice versa.  We then impose
%% a certain notion of equivalence on sequent calculus derivations in
%% order to prove that these two mappings are inverses of each other
%% (i.e., that derivations in the sequent calculus and maps in the free
%% skew monoidal category are in bijection).  Finally, we identify the
%% subcalculus of focused derivations, and show that every equivalence
%% class of derivations contains exactly one focused derivation. This
%% means that the focused sequent calculus characterizes the free skew
%% monoidal category in a particularly appealing fashion. Moreover, as
%% the focused representation of a derivation can be easily computed, we
%% get a simple algorithm for deciding equality of morphisms in the free
%% skew monoidal category: two morphisms are equal if and only if they
%% correspond to the same focused derivation.
%% Also, since focused
%% derivations can be systematically searched for, we get an algorithm
%% for enumerating (without duplications) any homset, in particular, for
%% deciding whether a homset is inhabited.

The approach that we take in this paper draws strong inspiration from Lambek's pioneering work applying proof-theoretic techniques towards category theory and vice versa.
These mutual influences can already be clearly seen in Lambek's original paper on the syntactic calculus (and even more so in an immediate followup article \cite{Lambek1961}), but the connection between proof theory and category theory was also explored explicitly in his papers on ``Deductive systems and categories'' \cite{Lambek1968,Lambek1969}.
In the latter of those two papers he formally introduced the concept of \emph{multicategory} (cf.~\cite{Lambek1989}), which is useful in better understanding the proof-theoretic analysis that we develop here.
Indeed, in recent work independent of ours, Bourke and Lack \cite{BL:multi} have related skew monoidal categories to what they call \emph{skew multicategories,} establishing a correspondence between skew monoidal categories and \emph{left representable} skew multicategories.
The two analyses are in fact closely related.
In a certain sense that we will make precise, the sequent calculus for skew monoidal categories can be seen as providing an explicit construction of the free left representable skew multicategory over a set of generators.

This paper is organized as follows. In Section~\ref{sec:skewmoncat}, we
review skew monoidal categories and present the free skew monoidal category
(over a set $\Var$ of generators that we view as atoms) as a simple deductive system whereby each map is an equivalence class of derivations wrt.\ a suitable equivalence relation.
In Section~\ref{sec:seqcalc}, we present the skew monoidal sequent calculus,
and show that it captures existence of maps in $\Fsk(\Var)$ via soundness and completeness theorems.
In Section~\ref{sec:adequacy}, we refine this correspondence to reflect
equality of maps in $\Fsk(\Var)$ by introducing an appropriate equivalence
relation on cut-free derivations.
In Section~\ref{sec:focusing}, we identify the focused subsystem of the sequent calculus providing canonical representatives of each equivalence class, and prove the above-mentioned coherence theorem(s) for the free skew monoidal category.
In Section~\ref{sec:bourke-lack}, we discuss the relationship to Bourke and Lack's work (and Lambek's) in more detail, also introducing an equivalent reformulation of skew multicategories inspired by the sequent calculus.
Finally, in Section~\ref{sec:concl-future}, we conclude and outline some future directions.
% In Appendix~\ref{sec:genmult} we show that the sequent
% calculus satisfies the laws of a generalized multicategory, and in
% Appendix~\ref{sec:skewmulticat} that it is in particular a left
% representable skew multicategory. Finally, in Appendix~\ref{sec:ineq},
% we discuss a partial order one can introduce on the categorical
% calculus derivations and the corresponding partial order on the
% sequent calculus derivations.

We have fully formalized the development of 
Sections~\ref{sec:skewmoncat}--\ref{sec:focusing} (and most of Section~\ref{sec:bourke-lack}) in the dependently
typed programming language Agda.
Our formalization is available at \url{http://cs.ioc.ee/~niccolo/skewmonseqcalc/}.
It is based on Agda 2.5.3 with standard library 0.14.

%% \subsection{Connection to Lambek's work}

%% a few threads:
%% \begin{itemize}
%% \item Lambek's original 1958 paper \cite{Lam:matss} introducing the \emph{syntactic calculus} (now called the Lambek calculus), which can be seen as the logic of monoidal biclosed categories (without unit in his original formulation).
%%   He proved a Gentzen-type cut-elimination theorem and explained how it can be used to obtain a decision procedure for derivability.
%% \item already in his followup 1961 paper \cite{Lambek1961} he discussed the idea of interpreting sequent calculus derivations (in the associative syntactic calculus) as multilinear mappings.  He refers to earlier unpublished work with Findlay (1955), as well as Bourbaki's \emph{Alg\`ebre multilin\'eaire} (1948).
%% \item these ideas are revisited in the papers on ``Deductive systems and categories'' \cite{Lambek1968,Lambek1969}. The 1969 paper he explicitly introduces the notion of \emph{multicategory}. Examines history in ``Multicategories revisited'' \cite{Lambek1989}.
%% \end{itemize}

\section{Skew Monoidal Categories}
\label{sec:skewmoncat}

% Szlach\'anyi's skew monoidal categories \cite{Szl:skemcb} are a variation of monoidal categories.

A \emph{skew monoidal category} \cite{Szl:skemcb} is a category $\C$ together
with a distinguished object $\I$, a functor $\ot : \C \times \C \to
\C$ and three natural transformations % $\lam$, $\rho$, $\al$ typed 
\[
\lam_{A} : \I \ot A \to A \qquad
\rho_{A} : A \to A \ot \I\qquad
\al_{A,B,C} : (A \ot B) \ot C \to A \ot (B \ot C)
% \begin{array}{c}
% \lam_{A} : \I \ot A \to A \\
% \rho_{A} : A \to A \ot \I\\
% \al_{A,B,C} : (A \ot B) \ot C \to A \ot (B \ot C)
% \end{array}
\]
satisfying the following laws:
\[
\mathrm{(a)}  
\xymatrix@R=1.5pc@C=0.2pc{
    & \I \ot \I \ar[dr]^-{\lambda_\I} & \\
    \I \ar[ur]^-{\rho_\I} \ar@{=}[rr] & & \I
    }
\qquad
\mathrm{(b)}\   
\xymatrix@R=1.3pc{
      (A\ot \I) \ot B \ar[r]^{\alpha_{A,\I,B}}
      & A \ot (\I\ot B) \ar[d]^{A\ot \lambda_{B}}\\
      A \ot B \ar@{=}[r] \ar[u]^{\rho_{A}\ot B}&  A \ot B 
    }
\]
\[
\mathrm{(c)}\ 
\xymatrix@C=0.2pc@R=1.5pc{
  (\I \ot A) \ot B \ar[dr]_{\lambda_A \ot B} \ar[rr]^{\alpha_{\I,A,B}} 
           & &  \I \ot (A \ot B) \ar[dl]^{\lambda_{A \ot B}}\\
  & A \ot B & 
    }
\quad
\mathrm{(d)}\ 
\xymatrix@C=0.2pc@R=1.5pc{
  (A \ot B) \ot \I \ar[rr]^{\alpha_{A,B,\I}} 
           & &  A \ot (B \ot \I) \\
  & A \ot B \ar[ul]^{\rho_{A \ot B}} \ar[ur]_{A \ot \rho_B} & 
    }
\]
\[
\mathrm{(e)}\ 
\xymatrix@R=1.5pc@C=2.2pc{
(A\ot (B \ot C)) \ot D \ar[rr]^{\alpha_{A,B \ot C,D}}
  & & A\ot ((B \ot C)\ot D) \ar[d]^{A\ot \alpha_{B,C,D}}
  \\
((A\ot B) \ot C) \ot D \ar[u]^{\alpha_{A,B,C} \ot D}
      \ar[r]^{\alpha_{A\ot B,C,D}}
  & (A\ot B) \ot (C \ot D) \ar[r]^{\alpha_{A,B,C\ot D}}
    & A\ot (B \ot (C \ot D))
}
\]
Notice that (a)--(e) are directed versions of the original Mac Lane
axioms \cite{ML:natac}. Kelly \cite{Kel:maclcc} observed that
when $\lam$, $\rho$ and $\al$ are natural isomorphisms, laws (a), (c),
and (d) can be derived from (b) and (e). However, for skew monoidal
categories, this is not the case.

\medskip

Skew monoidal categories arise more often than one would perhaps first
think, see \cite{Szl:skemcb,LS:skemsw,BGLS:catss,Uus:cohsmc}. The
following are some examples from \cite{Uus:cohsmc}.

\begin{nexample}
A simple example of a skew monoidal category results from
skewing a numerical addition monoid.

View the partial order $(\mathbb{N}, \leq)$ of natural numbers as a
thin category. Fix some natural number $n$ and define $\I = n$ and $x
\otimes y = (x \dotminus n) + y$ where $\dotminus$ is ``truncating
subtraction'' $a \dotminus b = \max(a-b,0)$. We have $\lam_x : (n \dotminus n) + x = 0 + x = x$,
$\rho_x : x \leq \max(x, n) = (x \dotminus n) + n$, and
$\al_{x,y,z} : (((x \dotminus n) + y) \dotminus n) + z \leq (x
\dotminus n) + (y \dotminus n) + z$ (the last fact by a small case analysis).
\end{nexample}

\begin{nexample}
The category $\Ptd$ of pointed sets and point-preserving functions has the
following skew monoidal structure. 

Take $\I = (1, \ast)$ and $(X, p) \otimes (Y, q) = (X + Y, \inl~p)$
(notice the ``skew'' in choosing the point). We define $\lam_X : (1,
\ast) \ot (X, p) = (1+X, \inl~\ast) \to (X, p)$ by $\lam_X~(\inl~\ast)
= p$, $\lam_X~(\inr~x) = x$ (this is not injective).  We let $\rho_X :
(X, p) \to (X + 1, \inl~p) =(X, p) \ot (1, \ast)$ by $\rho_X~x =
\inl~x$ (this is not surjective). Finally we let $\al_{X,Y,Z} : ((X, p)
\ot (Y, q)) \ot (Z, r) = ((X + Y) + Z, \inl~(\inl~p)) \to (X + (Y + Z),
\inl~p) = (X, p) \ot ((Y, q) \ot (Z, r))$ be the obvious isomorphism.

(We note that $\Ptd$ has coproducts too:
$(X, p) + (Y, q) = ((X + Y)/{\sim}, [\inl~p])$ where $\sim$ is the
equivalence relation on $X + Y$ induced by $\inl~p \sim \inr~q$.)
\end{nexample}

\begin{nexample}
Suppose given a monoidal category $(\C, \I, \ot)$ together with a lax
monoidal comonad $(D, \e, \m)$ on $\C$. The category $\C$ has a
skew monoidal structure with $\I^D = \I$, $A \ot^D B = A \ot
D\, B$. The unitors and associator are the following:
\begin{tabbing}
\quad\= $\lam^D_A = \xymatrix{\I \ot D\,A \ar[r]^-{\I \ot \varepsilon_A} & \I \ot A \ar[r]^-{\lambda_A} & A}$ \\
\> $\rho^D_A = \xymatrix{A \ar[r]^-{\rho_A} & A \ot \I \ar[r]^-{A \ot \e} & A \ot D\,\I}$ \\
\> $\al^D_{A,B,C} = \xymatrix{(A \ot D\,B) \ot D\,C 
                   \ar[rr]^-{(A \ot D\,B) \ot \delta_C}
                   && (A \ot D\,B) \ot D\,(D\,C)}$ \\
\> \ \qquad $\xymatrix{\ar[rr]^-{\al_{A,DB,D(DC)}}
                 && A \ot (D\,B \ot D\,(D\,C)) \ar[rr]^-{A \ot \m_{B,DC}}
                 && A \ot D\,(B \ot D\,C)}$
\end{tabbing}
%\]
A similar skew monoidal category is obtained from an oplax
monoidal monad.
\end{nexample}

\begin{nexample}
Consider two categories $\JJ$ and $\C$ with a functor $J : \JJ \to \C$,
and assume that the left Kan extension $\Lan_J\, F : \C \to \C$ exists
for every $F : \JJ \to \C$.
Then the functor category $[\JJ, \C]$ has a skew monoidal structure given
by $\I = J$, $F \otimes G = \mathrm{Lan}_J F \cdot G$.  The unitors and associator are the canonical natural
transformations $\lam_F : \Lan_J~J \cdot F \to F$, $\rho_F : F \to
\Lan_J~F \cdot J$, $\al_{F,G,H} : \Lan_J~(\Lan_J~F \cdot G) \cdot H
\to \Lan_J~F \cdot \Lan_J~G \cdot H$. This category becomes properly
monoidal under certain conditions on $J$: $\rho$ is an isomorphism if
$J$ is fully-faithful, and $\lam$ is an isomorphism if $J$ is
dense. (This is the example from our relative monads work
\cite{ACU:monnnb}. Relative monads on $J$ are skew monoids in the
skew monoidal category $[\JJ, \C]$.)
\end{nexample}

As our aim is to analyze the relationship of skew monoidal categories
to a sequent calculus with the methods of structural proof theory, we
will find it convenient to have an explicit description of the free
skew monoidal category $\Fsk(\Var)$ over a set $\Var$ as a deductive
calculus (in the style of Lambek \cite{Lambek1968,Lambek1969}); 
we refer to it as the skew monoidal \emph{categorical calculus}.

Objects of $\Fsk(\Var)$ are called \emph{formulae}, and are defined inductively
as follows: a formula is either an element $X$ of $\Var$ (an atomic formula); 
$\I$; or $A \ot B$ where $A$, $B$ are formulae.
We write $\Tm$ for the set of formulae.
% \[
% \infer{X : \Tm}{X : \Var}
% \quad
% \infer{\I : \Tm}{}
% \quad
% \infer{A \ot B : \Tm}{A : \Tm & B : \Tm}
% \]

Morphisms of $\Fsk(\Var)$ are derivations of
(singleton-antecedent, singleton-succedent) sequents $A \tto C$ where $A$, $C$ 
are formulae\footnote{In the sequent calculus to be introduced in the next 
section, we have a different form of sequents.} that are constructed 
using the rules of Figure~\ref{fig:catcalc} and are
identified up to the least congruence $\doteq$ given by the equations in
Figure~\ref{fig:skewmon}.  In addition to the laws (a)--(e) above,
these equations state that $\id$ and $\dcomp$ satisfy the laws of a
category, that $\otimes$ is functorial, and that $\lam$, $\rho$ and
$\al$ are natural transformations. In the term notation for
derivations, we write $g \comp f$ for $\dcomp\, f\, g$ to agree with
the standard categorical notation.

\begin{figure}
\[
\small
\renewcommand{\arraystretch}{2}
\begin{array}{c}
\infer[\id]{A \tto A}{}
\qquad
\infer[\dcomp]{A \tto C}{A \tto B & B \tto C}
%\\
\qquad
\infer[\ot]{A \ot B \tto C \ot D}{A \tto C & B \tto D}
\\
\infer[\lam]{\I \ot A \tto A}{}
\qquad
\infer[\rho]{A \tto A \ot \I}{}
\qquad
\infer[\al]{(A \ot B) \ot C \tto A \ot (B \ot C)}{}
% \infer{\id : A \tto A}{}
% \quad
% \infer{f \comp g : A \tto C}{f : B \tto C & g : A \tto B}
% \\
% \infer{f \ot g : A \ot B \tto C \ot D}{f : A \tto C & g : B \tto D}
% \\
% \infer{\lam : I \ot A \tto A}{}
% \quad
% \infer{\rho : A \tto A \ot I}{}
% \quad
% \infer{\al : (A \ot B) \ot C \tto A \ot (B \ot C)}{}
\end{array}
\]
\caption{Rules of the skew monoidal categorical calculus}
\label{fig:catcalc}
\end{figure}

% defined inductively by the rules
% \[
% \renewcommand{\arraystretch}{2}
% \begin{array}{c}
% % This are just the axioms of a congruence
% \infer{f \doteq f}{}
% \quad
% \infer{g \doteq f}{f \doteq g}
% \quad
% \infer{f \doteq h}{f \doteq g & g \doteq h}
% \quad
% \infer{f \comp h \doteq g \comp k}{f \doteq g & h \doteq k}
% \quad
% \infer{f \ot h \doteq g \ot k}{f \doteq g & h \doteq k}
% \\
% % Here the interesting equations begin
% \infer{\id \comp f \doteq f}{}
% \quad
% \infer{f \doteq f \comp \id}{}
% \quad
% \infer{(f \comp g) \comp h \doteq f \comp (g \comp h)}{}
% \\
% \infer{\id \ot \id \doteq \id}{}
% \quad 
% \infer{(h \comp f) \ot (k \comp g) \doteq h \ot k \comp f \ot g}{}
% \\
% \infer{\lam \comp \id \ot f \doteq f \comp \lam}{}
% \quad
% \infer{\rho \comp f \doteq f \ot \id \comp \rho}{}
% \quad
% \infer{\al \comp (f \ot g) \ot h \doteq f \ot (g \ot h) \comp \al}{}
% \\
% \infer{\lam \comp \rho \doteq \id}{}
% \quad
% \infer{\id \doteq \id \ot \lam \comp \al \comp \rho \ot \id}{}
% \\
% \infer{\lam \comp \al \doteq \lam \ot \id}{}
% \quad
% \infer{\al \comp \rho \doteq \id \ot \rho}{}
% \quad
% \infer{\al \comp \al \doteq \id \ot \al \comp \al \comp \al \ot \id}{}
% \end{array}
% \]

\begin{proposition}[cf.~Proposition 2 of \cite{Lambek1968}]
  Let $\C$ be any skew monoidal category equipped with a function $G : \Var \to \C$ interpreting atoms as objects of $\C$.
  Then $G$ extends uniquely to a strict monoidal functor $\bar G : \Fsk(\Var) \to \C$ compatible with the inclusion $\Var \to \Fsk(\Var)$.
\end{proposition}
%
%\afterpage{
\begin{sidewaysfigure}
  \vspace{0.5\textwidth}
%\begin{landscape}
{\tiny
\[
\begin{array}{c}
\proofbox{
\infer[\dcomp]{A \tto B}{
  A \tto[f] B &
  \infer[\id]{B \tto B}{}
}
}
\,\doteq\,
%\proofbox{
A \tto[f] B
%}
\hspace*{5mm}
%\proofbox{
A \tto[f] B
%}
\,\doteq\,
\proofbox{
\infer[\dcomp]{A \tto B}{
  \infer[\id]{A \tto A}{} & 
  A \tto[f] B
}
}
\hspace*{5mm}
\proofbox{
\infer[\dcomp]{A \tto D}{
  A \tto[f] B &
  \infer[\dcomp]{B \tto D}{
    B \tto[g] C &
    C \tto[h] D
  }    
}
}
\,\doteq\,
\proofbox{
\infer[\dcomp]{A \tto D}{
  \infer[\dcomp]{A \tto C}{
    A \tto[f] B &
    B \tto[g] C
  } &   
  C \tto[h] D
}
}
\\
\proofbox{
\infer[\ot]{A \ot B \tto A \ot B}{
  \infer[\id]{A \tto A}{} &
  \infer[\id]{B \tto B}{}
}
}
~ \doteq ~
\proofbox{
\infer[\id]{A \ot B \tto A \ot B}{
}
}
\hspace*{1cm}
\proofbox{
\infer[\ot]{A \ot B \tto E \ot F}{
  \infer[\dcomp]{A \tto E}{
    A \tto[f] C &
    C \tto[h] E
  } &
  \infer[\dcomp]{B \tto F}{
    B \tto[g] D &
    D \tto[k] F
  }
}
}
~ \doteq ~
\proofbox{
\infer[\dcomp]{A \ot B \tto E \ot F}{
  \infer[\ot]{A \ot B \tto C \ot D}{
    A \tto[f] C &
    B \tto[g] D
  } &   
  \infer[\ot]{C \ot D \tto E \ot F}{
    C \tto[h] E &
    D \tto[k] F
  } 
}
}
\\
\proofbox{
\infer[\dcomp]{\I \ot A \tto B}{
  \infer[\ot]{\I \ot A \tto \I \ot B}{
    \infer[\id]{\I \tto \I}{} &
    A \tto[f] B
  } &
  \infer[\lam]{\I \ot B \tto B}{}
}
}
~ \doteq ~
\proofbox{
\infer[\dcomp]{\I \ot A \tto B}{
  \infer[\lam]{\I \ot A \tto A}{} &
  A \tto[f] B    
}
}
\hspace*{1cm}
\proofbox{
\infer[\dcomp]{A \tto B \ot \I}{
  \infer[\rho]{A \tto A \ot \I}{} &
  \infer[\ot]{A \ot \I \tto B \ot \I}{
    A \tto[f] B &
    \infer[\id]{\I \tto \I}{} 
  }
}
}
~ \doteq ~
\proofbox{
\infer[\dcomp]{A \tto B \ot \I}{
  A \tto[f] B &
  \infer[\rho]{B \tto B \ot \I}{}
}
}
\\
\proofbox{
\infer[\dcomp]{(A \ot B) \ot C \tto D \ot (E \ot F)}{
  \infer[\al]{(A \ot B) \ot C \tto A \ot (B \ot C)}{} &
  \infer[\ot]{A \ot (B \ot C) \tto D \ot (E \ot F)}{
    A \tto[f] D &
    \infer[\ot]{B \ot C \tto E \ot F}{
      B \tto[g] E &
      C \tto[h] F &
    }
  }
}
}
~ \doteq ~
\proofbox{
\infer[\dcomp]{(A \ot B) \ot C \tto D \ot (E \ot F)}{
  \infer[\ot]{(A \ot B) \ot C \tto (D \ot E) \ot F}{
    \infer[\ot]{B \ot C \tto E \ot F}{
      A \tto[f] D &
      B \tto[g] E &      
    } &
  C \tto[h] F
  } &
  \infer[\al]{(D \ot E) \ot F \tto D \ot (E \ot F)}{}
}
}
\\
\proofbox{
\infer[\dcomp]{\I \tto \I}{
  \infer[\rho]{\I \tto \I \ot \I}{} &
  \infer[\lam]{\I \ot \I \tto \I}{}
}  
}
~ \doteq ~
\proofbox{
\infer[\id]{\I \tto \I}{}
}
\hspace*{1cm}
\proofbox{
\infer[\dcomp]{A \ot B \tto A \ot B}{
  \infer[\ot]{A \ot B \tto (A \ot \I) \ot B}{
    \infer[\rho]{A \tto A \ot I}{}
    &
    \infer[\id]{B \tto B}{}
  }  
  &
  \infer[\dcomp]{(A \ot \I) \ot B \tto A \ot B}{
    \infer[\al]{(A \ot \I) \ot B \tto A \ot (\I \ot B)}{} 
    &
    \infer[\ot]{A \ot (\I \ot B) \tto A \ot B}{
      \infer[\id]{A \tto A}{} 
       &
      \infer[\lam]{\I \ot B\tto B}{}
    }
  }
}    
}
\doteq ~
\proofbox{
\infer[\id]{A \ot B \tto A \ot B}{}
}
\\[12pt]
\proofbox{
\infer[\dcomp]{(\I \ot A) \ot B \tto A \ot B}{
  \infer[\al]{(\I \ot A) \ot B \tto \I \ot (A \ot B)}{
  }
  &
  \infer[\lam]{\I \ot (A \ot B) \tto A \ot B}{
  }
}
}
~ \doteq ~
\proofbox{
\infer[\ot]{(\I \ot A) \ot B \tto A \ot B}{
  \infer[\lam]{\I \ot A \tto A}{
  }
  &
  \infer[\id]{B \tto B}{
  }
}
}
\\[12pt]
\proofbox{
\infer[\dcomp]{A \ot B \tto A \ot (B \ot I)}{
  \infer[\rho]{A \ot B \tto (A \ot B) \ot \I}{
  }
  &
  \infer[\al]{(A \ot B) \ot \I \tto  A \ot (B \ot I)}{
  }
}
}
~ \doteq ~
\proofbox{
\infer[\ot]{A \ot B \tto A \ot (B \ot I)}{
  \infer[\id]{A \tto A}{
  }
  &
  \infer[\rho]{B \tto B \ot \I}{
  }
}
}
\\[12pt]
\proofbox{
\infer[\dcomp]{((A \ot B) \ot C) \ot D \tto A \ot (B \ot (C \ot D))}{
  \infer[\al]{((A \ot B) \ot C) \ot D \tto (A \ot B) \ot (C \ot D)}{
  }
  &
  \infer[\al]{(A \ot B) \ot (C \ot D) \tto A \ot (B \ot (C \ot D))}{
  }
}
}
\hspace*{10cm} \\
\hspace*{3cm}
~ \doteq ~
\proofbox{
\infer[\dcomp]{((A \ot B) \ot C) \ot D \tto A \ot (B \ot (C \ot D))}{
  \infer[\ot]{((A \ot B) \ot C) \ot D \tto (A \ot (B \ot C)) \ot D}{
    \infer[\al]{(A \ot B) \ot C \tto A \ot (B \ot C)}{}
    &
    \infer[\id]{D \tto D}{}
  }  
  &
  \infer[\dcomp]{(A \ot (B \ot C)) \ot D \tto A \ot (B \ot (C \ot D))}{
    \infer[\al]{(A \ot (B \ot C)) \ot D \tto A \ot ((B \ot C) \ot D)}{} 
    &
    \infer[\ot]{A \ot ((B \ot C) \ot D) \tto A \ot (B \ot (C \ot D))}{
      \infer[\id]{A \tto A}{} 
       &
      \infer[\al]{(B \ot C) \ot D \tto B \ot (C \ot D)}{}
    }
  }
}    
}
\end{array}
\]}
\caption{Equations on the categorical calculus derivations}
\label{fig:skewmon}
%\end{landscape}
\end{sidewaysfigure}
%}

We make here some simple observations about maps in $\Fsk(\Var)$.
Let $\partial A$ denote the underlying list of atoms forming the frontier of a formula $A$, where $\partial X = X, \partial \I = ()$ and $\partial(A\ot B) = \partial A,\partial B$.
A necessary condition for the existence of a map $A \tto B$ in $\Fsk(\Var)$ is that $\partial A = \partial B$, but this is not sufficient. For example, there are no maps
$X \tto \I \ot X$, $X \ot \I \tto X$ or
$X \ot (Y \ot Z) \tto (X \ot Y) \ot Z$ in the free skew monoidal category, although these exist in any monoidal category as the inverses of the
unitors and the associator. 
Moreover, it is possible to have more than one map with the same 
domain and codomain: the prototypical examples
are
$\rho_\I \comp \lam_\I \not\doteq \id_{\I \ot \I}$, % : \I \ot \I \tto \I \ot \I$,
$\id_{(X\ot \I) \ot Y} \not\doteq (\rho_X \ot \lam_Y) \comp \al_{X,\I,Y}$
and
$\id_{X\ot (\I \ot Y)} \not\doteq \al_{X,\I,Y} \comp (\rho_X \ot \lam_Y)$.
(In contrast, all of these equations hold in any monoidal category.)
%, as variations of equations (a) and (b).)

\section{A Skew Monoidal Sequent Calculus}
\label{sec:seqcalc}

We now introduce the \emph{sequent calculus} for skew monoidal categories
inspired by the sequent calculus for the Tamari order \cite{Zei}.

The inference rules of this sequent calculus are given in Figure~\ref{fig:seqcalc}, with the standard ``monoidal sequent calculus'' included for comparison in Figure~\ref{fig:lambek}.
(The latter corresponds exactly to the division-free fragment of Lambek's original syntactic calculus \cite{Lam:matss} extended with a unit $\I$, as considered for example in \cite{Lambek1969}.)
Sequents of this calculus are of the form $\stseq{S}{\Gamma}{C}$, where
the antecedent is a pair of a \emph{stoup} $S$ together with a \emph{context} $\Gamma$ and the succedent $C$ is a single formula.
A stoup can be either empty (written $S = \n$) or contain a single formula, while a context is an arbitrary-length list of formulae.%
\begin{figure}[t]
\[
\infer[\ax]{\mseq{A}{A}}{}
\qquad
\infer[\cut]{\mseq{\Delta_0, \Gamma, \Delta_1}{C}}{
  \mseq{\Gamma}{A}
  & 
  \mseq{\Delta_0, A, \Delta_1}{C}
}
\]

\[
\infer[\IL]{\mseq{\Gamma_0,\I,\Gamma_1}{C}}{
  \mseq{\Gamma_0,\Gamma_1}{C}
}
\quad
\infer[\IR]{\mseq{~}{\I}}{
}
\quad
\infer[\otL]{\mseq{\Gamma_0, A\ot B, \Gamma_1}{C}}{
  \mseq{\Gamma_0, A, B, \Gamma_1}{C}
}
\quad
\infer[\otR]{\mseq{\Gamma, \Delta}{A \otimes B}}{
  \mseq{\Gamma}{A}
  &
  \mseq{\Delta}{B}
}
\]
\caption{Rules of the (ordinary) monoidal sequent calculus (cf.~\cite{Lam:matss,Lambek1969})}
\label{fig:lambek}
\end{figure}
\begin{figure}[t]
\[
\infer[\ax]{\stseq{A}{~}{A}}{
}
\qquad
\infer[\uf]{\stseq{\n}{A, \Gamma}{C}}{
  \stseq{A}{\Gamma}{C}
}
\]

\[
\infer[\scut]{\stseq{S}{\Gamma, \Delta}{C}}{
  \stseq{S}{\Gamma}{A}
  & 
  \stseq{A}{\Delta}{C}
}
\qquad
\infer[\ccut]{\stseq{S}{\Delta_0, \Gamma, \Delta_1}{C}}{
  \stseq{\n}{\Gamma}{A}
  & 
  \stseq{S}{\Delta_0, A, \Delta_1}{C}
}
\]
\\[-8pt]
\[
\infer[\IL]{\stseq{\I}{\Gamma}{C}}{
  \stseq{\n}{\Gamma}{C}
}
\quad
\infer[\IR]{\stseq{\n}{~}{\I}}{
}
\quad
\infer[\otL]{\stseq{A \ot B}{\Gamma}{C}}{
  \stseq{A}{B, \Gamma}{C}
}
\quad
\infer[\otR]{\stseq{S}{\Gamma, \Delta}{A \otimes B}}{
  \stseq{S}{\Gamma}{A}
  &
  \stseq{\n}{\Delta}{B}
}
\]
\caption{Rules of the skew monoidal sequent calculus}
\label{fig:seqcalc}
\end{figure}
Before considering some examples, we highlight a few important properties of the calculus:
\begin{enumerate}
\item As in the Lambek calculus (but in contrast to Gentzen's original sequent calculi for classical and intuitionistic logic \cite{Gentzen35}), there are no structural rules of exchange, weakening or contraction.
\item As in the sequent calculus for the Tamari order \cite{Zei}, the left logical rules are restricted to apply only at the leftmost end of the antecedent, specifically to the formula within the stoup.
\item Finally (and this is a new aspect), the stoup
is allowed to be empty, permitting a distinction between antecedents
of the form $A \mid \Gamma$ (with $A$ inside the stoup) and
antecedents of the form $\n \mid A, \Gamma$ (with $A$ outside the
stoup).
The $\uf$ rule is used to move from one form of antecedent to another, while there are now two forms of cut rule, one for substitution into the stoup ($\scut$), and one for substitution into the context ($\ccut$).
\end{enumerate}

A consequence of all these restrictions will be the following:
\begin{nclaim}\label{claim:adequacy}
  $A \tto C$ is derivable in the categorical calculus (i.e, there exists a map $A \tto C$ in $\Fsk(\Var)$) iff $\stseq{A}{~}{C}$ is derivable in the sequent calculus.
\end{nclaim}
We will prove Claim~\ref{claim:adequacy} at the end of this section by providing
effective translations between the two calculi. In Section~\ref{sec:adequacy}, we will strengthen it by showing, on the one hand, that the translations respect equivalence of derivations (we will introduce an equivalence relation also on the derivations of the sequent calculus) and, on the other hand, that they are inverses.

For now, let us demonstrate the calculus in action on a few examples.

As a first example, here is a complete derivation corresponding to the skew associator $\al_{X,Y,Z} : (X \ot Y) \ot Z \tto X \ot (Y \ot Z)$:
\begin{equation}\label{deriv:al}
\small
\vcenter{
\infer[\otL]{\stseq{(X \ot Y) \ot Z}{~}{X \ot (Y \ot Z)}}{
  \infer[\otL]{\stseq{X \ot Y}{Z}{X \ot (Y \ot Z)}}{
    \infer[\otR]{\stseq{X}{Y, Z}{X \ot (Y \ot Z)}}{
      \infer[\ax]{\stseq{X}{~}{X}}{} &
      \infer[\uf]{\stseq{\n}{Y, Z}{Y \ot Z}}{
        \infer[\otR]{\stseq{Y}{Z}{Y \ot Z}}{
          \infer[\ax]{\stseq{Y}{~}{Y}}{}
          &
          \infer[\uf]{\stseq{\n}{Z}{Z}}{
            \infer[\ax]{\stseq{Z}{~}{Z}}{}
          }
        }
      }
    }
  }
}
}
\end{equation}
In the monoidal sequent calculus, one can also build a derivation corresponding to the inverse associator $\al^{-1}_{X,Y,Z}$: %  : X \ot (Y \ot Z) \tto (X \ot Y) \ot Z
\[
\small
\infer[\otL]{\mseq{X \ot (Y \ot Z)}{(X \ot Y) \ot Z}}{
  \infer[\otL]{\mseq{X, Y \ot Z}{(X \ot Y) \ot Z}}{
    \infer[\otR]{\mseq{X, Y, Z}{(X \ot Y) \ot Z}}{
      \infer[\otR]{\mseq{X, Y}{X \ot Y}}{
        \infer[\ax]{\mseq{X}{X}}{
        }   
        &
        \infer[\ax]{\mseq{Y}{Y}}{ 
        }
      }
      &
      \infer[\ax]{\mseq{Z}{Z}}{
      }
    }
  }
}
\]
Notice, however,
that the second application of the $\otL$ rule from the bottom is
to the formula second from the left ($Y \ot Z$) in the antecedent.
Such an application of the $\otL$ rule is invalid for the skew monoidal calculus, and indeed the
corresponding sequent is not derivable: 
\[
\small
\infer[\otL]{\stseq{X \ot (Y \ot Z)}{~}{(X \ot Y) \ot Z}}{
  \infer[\otR]{\stseq{X}{Y \ot Z}{(X \ot Y) \ot Z}}{
   \stseq[??]{X}{Y \ot Z}{X \ot Y}
   &
   \stseq[??]{\n}{~}{Z}
 }
}
\ \quad
\infer[\otL]{\stseq{X \ot (Y \ot Z)}{~}{(X \ot Y) \ot Z}}{
  \infer[\otR]{\stseq{X}{Y \ot Z}{(X \ot Y) \ot Z}}{
   \stseq[??]{X}{~}{X \ot Y}
   &
   \stseq[??]{\n}{Y \ot Z}{Z}
 }
}
\]
Here we have shown two incomplete attempts at building a derivation starting from the root, and both fail because there is no good way to split the context ($Y\ot Z$) between the two premises of the $\otR$ rule.
Of course, just seeing that these two attempts fail does not allow us to infer that \emph{all} proof attempts will fail, but indeed this is an immediate consequence of the focusing completeness theorem for the sequent calculus (a strengthening of cut-elimination), which we will prove in Section~\ref{sec:focusing}.
%: see Proposition~\ref{prop:admits-cut}, as well as the stronger focusing completeness result we will discuss in Section~\ref{sec:focusing}.)

%% (To pass from these failed attempts to the conclusion that the sequent has no derivation we are implicitly relying on the fact that the calculus admits cut-elimination: this is Proposition~\ref{prop:admits-cut} below.
%% We will have more to say about proof search in Section~\ref{sec:focusing}.)

As another similar example, the sequent corresponding to the right unitor $\rho_X : X \tto X \ot \I$ is derivable, but the converse sequent is not:
\begin{equation}\label{deriv:rho}
\small
\vcenter{
\infer[\otR]{\stseq{X}{~}{X \ot \I}}{
  \infer[\ax]{\stseq{X}{~}{X}}{} &
  \infer[\IR]{\stseq{\n}{~}{\I}}{}
}}
\qquad
\vcenter{
\infer[\otL]{\stseq{X \ot \I}{~}{X}}{
  \stseq[??]{X}{\I}{X}
}}
\end{equation}
By contrast, both directions are derivable in the monoidal sequent calculus, where the derivation of $\rho_X^{-1}$ is completed by applying $\IL$ to the second formula in an antecedent.
%% \[
%% \infer[\otL]{A \ot \I \vdash A}{
%%   \infer[\IL]{A, \I \vdash A}{
%%     \infer[\ax]{A \vdash A}{}
%%   }
%% }
%% \]

Likewise, the sequent corresponding to the left unitor $\lam_X : \I \ot X \tto X$ is derivable in the sequent calculus for skew monoidal categories, but its inverse is not:
\begin{equation}\label{deriv:lam}
\small
\vcenter{
\infer[\otL]{\stseq{\I \ot X}{~}{X}}{
  \infer[\IL]{\stseq{\I}{X}{X}}{
    \infer[\uf]{\stseq{\n}{X}{X}}{
      \infer[\ax]{\stseq{X}{~}{X}}{}
    }
  }
}}
\qquad
\vcenter{
\infer[\otR]{\stseq{X}{}{\I \ot X}}{
  \stseq[??]{X}{~}{\I}
  &
  \stseq[??]{\n}{~}{X}
}}
\end{equation}
Here the reason why the attempt at a derivation of $\lam^{-1}_X$ fails is that, although the context can be split freely in an $\otR$ inference, the stoup formula must go to the first premise.
By contrast, both directions are derivable in the monoidal sequent calculus.
%% \[
%% \infer[\otR]{A \vdash \I \ot A}{
%%   \infer[\IR]{~ \vdash \I}{
%%   }
%%   &
%%   \infer[\ax]{A \vdash A}{
%%   }
%% }
%% \]

%% At the same time, derivations corresponding to $\lam_A$, $\rho_A$,
%% $\al_{A,B,A}$ can be smoothly constructed in our calculus despite the
%% restrictions. They are needed and appear in the proof of
%% Theorem~\ref{thm:cmplt} below.

As a more involved example, here is a derivation corresponding to the
(incidentally) unique map
$(X \ot (\I\ot Y))\ot Z \tto (X \ot \I) \ot (Y \ot Z)$ in the free skew
monoidal category:
\[\footnotesize
\infer[\otL]{\stseq{(X \ot (\I\ot Y))\ot Z}{~}{(X \ot \I) \ot (Y \ot Z)}}{
\infer[\otL]{\stseq{X\ot (\I\ot Y)}{Z}{(X \ot \I) \ot (Y \ot Z)}}{
\infer[\otR]{\stseq{X}{\I\ot Y,Z}{(X \ot \I) \ot (Y \ot Z)}}{
  \infer[\otR]{\stseq{X}{~}{X \ot \I}}{
   \infer[\ax]{\stseq{X}{~}{X}}{} &
   \infer[\IR]{\stseq{\n}{~}{\I}}{}} &
  \infer[\uf]{\stseq{\n}{\I\ot Y,Z}{Y\ot Z}}{
  \infer[\otL]{\stseq{\I\ot Y}{Z}{Y\ot Z}}{
  \infer[\IL]{\stseq{\I}{Y,Z}{Y\ot Z}}{
  \infer[\otR]{\stseq{\n}{Y,Z}{Y\ot Z}}{
    \infer[\uf]{\stseq{\n}{Y}{Y}}{
    \infer[\ax]{\stseq{Y}{~}{Y}}{}} & 
    \infer[\uf]{\stseq{\n}{Z}{Z}}{
    \infer[\ax]{\stseq{Z}{~}{Z}}{}}}}}}}}}
\]
The reader is invited to check that, in contrast, there is no derivation
of the converse sequent, although the sequent
$\mseq{(X \ot \I) \ot (Y \ot Z)}{(X \ot (\I\ot Y))\ot Z}$ is derivable
in the monoidal sequent calculus.
%, and the isomorphism $(X \ot (\I\ot Y))\ot Z \cong (X \ot \I) \ot (Y \ot Z)$ h%olds in any ordinary monoidal category.
\medskip

We dedicate the remainder of this section to the proof of Claim~\ref{claim:adequacy}, which can be split into separate claims of soundness and completeness.
Completeness is the easier direction, so we show that first.

\begin{theorem}[Completeness]\label{thm:cmplt}
  For any derivation $f : A \tto C$ in the skew monoidal categorical calculus, 
  there is a derivation
  $\cmplt\,f: \stseq{A}{~}{C}$ in the skew monoidal sequent calculus.
\end{theorem}

\begin{proof}
  This follows from the fact that all of the rules of the skew monoidal categorical calculus (Figure~\ref{fig:catcalc}) are \emph{derived rules} of the skew monoidal sequent calculus under the interpretation of $A \tto C$ as $\stseq{A}{~}{C}$.
  Identity $\id$ and composition $\dcomp$ are translated directly to $\ax$ and $\scut$, respectively, while the $\ot$ rule is derived as follows:

\[
\small
\infer[\otL]{\stseq{A \ot B}{~}{C \ot D}}{
  \infer[\otR]{\stseq{A}{B}{C \ot D}}{
    \stseq{A}{~}{C} &
    \infer[\uf]{\stseq{\n}{B}{D}}{
      \stseq{B}{~}{D}
    }
  }
}
\]
Finally, the structural maps $\al_{A,B,C}$, $\rho_A$, and $\lam_A$ can all be derived as we have already seen, as derivations \eqref{deriv:al}, \eqref{deriv:rho}, and \eqref{deriv:lam} above, with atoms replaced with general formulae.
\end{proof}

In order to prove soundness, we first explain how to interpret general sequents $\stseq{S}{\Gamma}{C}$.
This is performed in several steps, by first specifying how to interpret stoups and contexts.

Stoups are interpreted as formulae by reading the empty stoup as the unit:
\[
\ssem{\n} \eqdf \I
\hspace*{2cm}
\ssem{A} \eqdf A
\]
Contexts are interpreted as a right action on formulae, given by iterated application of the tensor product $\ot$:
\[
\csem{C}{~} \eqdf C
\quad \quad
\csem{C}{A, \Gamma} \eqdf \csem{(C \ot A)}{\Gamma}
\]
Combining these two interpretations, we define the
interpretation of antecedents as formulae by applying the action of the context to the stoup:
\[
\asem{S \mid \Gamma} \eqdf \csem{\ssem{S}}{\Gamma}
\]
Explicitly, for any context $\Gamma = A_1,\dots,A_n$ and any non-empty stoup $A$ we have
$\asem{A \mid \Gamma} = (\dots((A \ot A_1) \ot A_2)\dots) \ot A_n$,
while for the empty stoup we have 
$\asem{\n \mid \Gamma} = (\dots((\I \ot A_1) \ot A_2)\dots) \ot A_n$.

The proof of soundness now relies on several simple properties of this interpretation of antecedents.

% The first lemma states that, for all contexts $\Gamma$, the map
% sending a formula $A$ to the formula $\llbracket A \mid
% \rrbracket_\mathsf{A}$ is a functor

% 
% The stoup is needed for correctly interpreting antecedents in the categorical calculus. 
% 
% 
% Interpretation of contexts:
% $C~[~] = C$; $C~[A, \Gamma] = (C \otimes A)~[\Gamma]$
% 
% Interpretation of antecedents:
% $[S \mid \Gamma] = [S]~[\Gamma]$
% 
% Interpretation of sequents:
% $[\Theta \vdash A] = [\Theta] \tto A$
% 

\begin{lemma}\label{lem:theta}
  For any stoup $S$ and contexts $\Gamma$, $\Delta$, $\asem{S \mid \Gamma, \Delta} = \asem{\asem{S\mid \Gamma} \mid \Delta}$.
  %% there is a map $\theta_{S,\Gamma,\Delta} : \asem{S \mid \Gamma, \Delta} \tto
  %% \asem{\asem{S\mid \Gamma} \mid \Delta}$.
%   , 
% $\theta^{-1}_{S,\Gamma,\Delta} : \asem{\asem{S\mid \Gamma} \mid
%   \Delta}\tto \asem{S \mid \Gamma, \Delta}$.
% 
\end{lemma}
\begin{proof}
  Immediate by induction on $\Gamma$.
%% It is sufficient to construct $\theta'_{A,\Gamma,\Delta} : \csem{A}{\Gamma, \Delta} \tto \csem{(\csem{A}{\Gamma})}{\Delta}$
%% for any formula $A$, and define $\theta_{S,\Gamma,\Delta} \eqdf
%% \theta'_{\ssem{S},\Gamma,\Delta}$. We proceed by induction on
%% $\Gamma$. If $\Gamma$ is empty, we take $\theta'_{A,(~),\Delta} \eqdf
%% \id$. If $\Gamma = C,\Gamma'$, then we take
%% $\theta'_{A,(C,\Gamma'),\Delta} \eqdf \theta'_{A \ot C,\Gamma,\Delta}$.
\end{proof}
\begin{lemma}\label{lem:actfun}
  For any derivation $f : A \tto B$ and context $\Gamma$, there is a
  derivation
  $\asem{f \mid \Gamma} : \asem{A \mid \Gamma} \tto \asem{B \mid
    \Gamma}$.
\end{lemma}
\begin{proof}
  We proceed by induction on $\Gamma$. If $\Gamma$ is empty, then we
  take $\asem{f \mid ~} \eqdf f$. If
  $\Gamma = C,\Gamma'$, then we take
  $\asem{f \mid C , \Gamma'} \eqdf \asem{f \ot \id_C \mid \Gamma'}$.
\end{proof}

\begin{lemma}\label{lem:psi}
  For any formulae $A$, $B$ and context $\Gamma$, there is a derivation
$\psi_{A,B,\Gamma} : \asem{A \ot B \mid \Gamma} \tto A \ot
  \asem{B \mid \Gamma}$.
\end{lemma}
\begin{proof}
We proceed by induction on $\Gamma$. If $\Gamma$ is empty, then we
take $\psi_{A,B,(~)} \eqdf \id_{A\ot B}$. If $\Gamma = C,\Gamma'$, then we take $\psi_{A,B,(C,\Gamma')}
\eqdf \psi_{A,B\ot C,\Gamma'} \comp \asem{\al \mid \Gamma'}$.
\end{proof}

\begin{lemma}\label{lem:varphi}
  For any stoup $S$ and contexts $\Gamma$, $\Delta$, there is
  a derivation
  $\varphi_{S,\Gamma,\Delta} : \asem{S \mid \Gamma , \Delta} \tto
  \asem{S \mid \Gamma} \ot \asem{\n \mid \Delta}$.
\end{lemma}
\begin{proof}
It is sufficient to construct $\varphi'_{A,\Gamma,\Delta} : \csem{A}{\Gamma , \Delta} \tto \csem{A}{\Gamma} \ot \asem{\n \mid
\Delta}$ for any formula $A$, and define
$\varphi_{S,\Gamma,\Delta} \eqdf \varphi'_{\ssem{S},\Gamma,\Delta}$. We proceed by induction on
$\Gamma$.  If $\Gamma$ is empty,
% then we have to construct $\varphi'_{A,(~),\Delta}: \asem{A \mid \Delta} \tto A \ot \asem{\n \mid \Delta}$.
we take $\varphi'_{A,(~),\Delta} \eqdf \psi_{A,\I,\Delta} \comp \asem{\rho \mid \Delta}$ making use Lemma~\ref{lem:psi}.
If $\Gamma = C,\Gamma'$, then we take $\varphi'_{A,(C,\Gamma'),\Delta} \eqdf \varphi'_{A\ot C,\Gamma',\Delta}$.
\end{proof}

\begin{theorem}[Soundness]\label{thm:sound}
  For any derivation $f : \stseq{S}{\Gamma}{C}$ in the skew monoidal sequent calculus, there is a
  derivation
  $\sound\, f: \asem{S \mid \Gamma} \tto C$ in the skew monoidal categorical calculus. As a special case, for any derivation $\stseq{A}{~}{C}$ there is a derivation $A \tto C$.
%  \nz{How about using linear notation for the constructed maps in $\Fsk$ to make this proof more compact?}
\end{theorem}
\begin{proof}
  We proceed by induction on $f$.
\begin{itemize}
\setlength\itemsep{1em}
\item Case $f = \infer[\ax]{\stseq{C}{~}{C}}{}$.  
%  In particular, we have $S = A = B$ and $\Gamma$ is empty, with $\asem{S \mid \Gamma} = A$.
  We take $\sound\,f \eqdf \id : C \tto C$. %define: %$\sound\,ax \eqdf \id$.

%% \[
%% \small
%% \begin{array}{c}
%% \sound\,
%%   \left( \proofbox{
%%     \infer[\ax]{C \mid ~\vdash C}{
%%     }
%%   }
%%   \right) 
%% ~ \eqdf  ~
%% \proofbox{
%% \infer={\asem{C \mid ~} \tto C}{
%%   \infer[\id]{C \tto C}{}
%%   }
%% }
%% \end{array}
%% \]

\item
  Case $f = \infer[\scut]{\stseq{S}{\Gamma',\Delta}{C}}{\stseq[f']{S}{\Gamma'}{A} & \stseq[g]{A}{\Delta}{C}}$.
%  Case $f = \scut(f',g)$, for some $f' : S \mid \Gamma' \vdash A'$ and $g : A' \mid \Delta \vdash B$.
  %  In particular, we have $\Gamma = (\Gamma',\Delta)$.
  There exist $\sound\,f' : \asem{S \mid \Gamma'} \tto A$ and $\sound\,g : \asem{A \mid \Delta} \tto C$ by the induction hypothesis, and therefore $\asem{\sound\,f' \mid \Delta} : \asem{S \mid \Gamma',\Delta} = \asem{\asem{S \mid \Gamma'} \mid \Delta} \tto \asem{A \mid \Delta}$ by Lemmata~\ref{lem:theta} \& \ref{lem:actfun}.
  We take $\sound\,f \eqdf \sound\,g \comp \asem{\sound\,f' \mid \Delta}$.

\item Case $f = \infer[\ccut]{\stseq{S}{\Delta_0,\Gamma',\Delta_1}{C}}{\stseq[f']{\n}{\Gamma'}{A} & \stseq[g]{S}{\Delta_0,A,\Delta_1}{C}}$.
%  Case $f = \ccut(f',g)$, for some $f' : \n \mid \Gamma' \vdash A'$ and $g : S \mid \Delta_0,A',\Delta_1 \vdash B$.
  %  In particular, we have $\Gamma = (\Delta_0,\Gamma',\Delta_1)$.
  Similarly to the previous case,
  we take $\sound\,f \eqdf \sound\,g \comp \asem{h \mid \Delta_1}$, where
  $h = (\id\ot\sound\,f') \comp \varphi_{S,\Delta_0,\Gamma'} : \asem{S \mid \Delta_0,\Gamma'} \tto \asem{S\mid \Delta_0}\ot A$ is defined using Lemma~\ref{lem:varphi}.

\item Case $f = \infer[\uf]{\stseq{\n}{A,\Gamma'}{C}}{\stseq[f']{A}{\Gamma'}{C}}$.
  %% Case $f \eqdf \uf\, f'$, for some $f' : A \mid \Gamma' \vdash C$.
  %% In particular, $S = \n$ and $\Gamma = A,\Gamma'$. We define:
  We take $\sound\,f \eqdf \sound\,f' \comp \asem{\lam\mid \Gamma'}$.
%% \[
%% \small
%% \begin{array}{c}
%%   \sound\,
%%   \left( \proofbox{
%%     \infer[\uf]{\n \mid A,\Gamma' \vdash C}{
%%       \infer*[f']{A \mid \Gamma' \vdash C}{}
%%     }
%%   }
%%   \right) 
%% ~ \eqdf ~
%% \proofbox{
%% \infer={\asem{\n \mid A, \Gamma'} \tto C}{
%% \infer[\dcomp]{\asem{\I \ot A \mid \Gamma'} \tto
%%   C}{
%%   \infer[\asem{\lam \mid \Gamma'}]{\asem{\I \ot A \mid \Gamma'} 
%%      \tto \asem{A \mid \Gamma'}}{} 
%%   &
%%   \infer*[\sound \,f']{\asem{A \mid \Gamma'} \tto C}{}
%% }
%% }
%% }
%% \end{array}
%% % \sound\,\left(
%% %   \infer*[f']{A \mid \Gamma' \vdash C}{}
%% %   \right)
%% \]

\item Case $f = \infer[\IL]{\stseq{\I}{\Gamma}{C}}{\stseq[f']{\n}{\Gamma}{C}}$ or $f = \infer[\otL]{\stseq{A\ot B}{\Gamma}{C}}{\stseq[f']{A}{B,\Gamma}{C}}$.
  %  Case $f \eqdf \IL\,f'$, for some $f' : \n \mid \Gamma \vdash C$. In particular, $S = \I$.
  In either case we take $\sound\,f\eqdf \sound\,f'$, since the interpretation of the premise sequent is equal to the interpretation of the conclusion.

%% \[
%% \small
%% \begin{array}{c}
%% \sound\,
%%   \left( \proofbox{
%%     \infer[\IL]{\I \mid \Gamma \vdash C}{
%%       \infer*[f']{\n \mid \Gamma \vdash C}{}
%%     }
%%   }
%%   \right) 
%% ~ \eqdf ~ %\sound\,f'
%% \proofbox{
%% \infer={\asem{\I \mid \Gamma} \tto C}{
%%   \infer*[\sound\,f']{\asem{\n \mid \Gamma} \tto C}{}
%% }
%% }
%% \end{array}
%% \]

%% \item Case $f \eqdf \otL\, f'$, for some $f' : A \mid B,\Gamma
%%   \vdash C$. In particular, $S = A \ot B$. We define:

%% \[
%% \small
%% \begin{array}{c}
%% \sound\,
%%   \left( \proofbox{
%%     \infer[\otL]{A \ot B \mid \Gamma \vdash C}{
%%       \infer*[f']{A \mid B,\Gamma \vdash C}{}
%%     }
%%   }
%%   \right) 
%% ~ \eqdf ~
%% %\sound\,f'
%% \proofbox{
%% \infer={\asem{A \ot B \mid \Gamma} \tto C}{
%%   \infer*[\sound\,f']{\asem{A \mid B,\Gamma} \tto C}{}
%% }
%% }
%% \end{array}
%% % \sound\,\left(
%% %   \infer*[f']{A \mid B,\Gamma \vdash C}{}
%% %   \right) &
%% % \llbracket  A \mid B,\Gamma \rrbracket_\mathsf{A} \tto C
%% \]

\item Case $f = \infer[\IR]{\stseq{\n}{~}{\I}}{}$.
  We take $\sound\,f\eqdf \id : \I \tto \I$.
%  Case $f \eqdf \IR$. In particular, $S = \n$, $A = \I$ and $\Gamma$ is empty.

%% \[
%% \small
%% \begin{array}{c}
%% \sound\,
%%   \left( \proofbox{
%%     \infer[\IR]{\n \mid ~\vdash I}{
%%     }
%%   }
%%   \right) 
%% ~ \eqdf ~
%% \proofbox{
%% \infer={\asem{\n \mid ~} \tto \I}{
%%   \infer[\id]{\I \tto \I}{}
%% }
%% }
%% \end{array}
%% \]

\item Case $f = \infer[\otR]{\stseq{S}{\Gamma_1,\Gamma_2}{C_1 \ot C_2}}{\stseq[f_1]{S}{\Gamma_1}{C_1} & \stseq[f_2]{\n}{\Gamma_2}{C_2}}$.
  %% Case $f \eqdf \otR\, f_1\, f_2$, for some $f_1 : S \mid
  %% \Gamma_1 \vdash C_1$ and $f_2 : \n \mid \Gamma_2 \vdash C_2$. In
  %% particular $\Gamma = \Gamma_1,\Gamma_2$ and $C = C_1 \ot C_2$.
  Take $\sound\,f\eqdf (\sound\,f_1\ot \sound\,f_2) \comp \varphi_{S,\Gamma_1,\Gamma_2}$.
%% \begin{tabbing}
%% \small
%% $\sound\,
%%   \left( \proofbox{
%%     \infer[\otR]{S \mid \Gamma_1,\Gamma_2 \vdash C_1 \ot C_2}{
%%       \infer*[f_1]{S \mid \Gamma_1 \vdash C_1}{} 
%%       &
%%       \infer*[f_2]{\n \mid \Gamma_2 \vdash C_2}{}
%%     }
%%   }
%%   \right) $
%% \\
%% $\eqdf ~
%% \proofbox{
%%   \infer[\dcomp]{\asem{S \mid \Gamma_1,\Gamma_2} \tto C_1 \ot C_2}{
%%     \infer[\varphi_{S,\Gamma_1,\Gamma_2}]{\asem{S \mid \Gamma_1,\Gamma_2}
%%        \tto \asem{S \mid \Gamma_1} \ot \asem{\n \mid \Gamma_2}}{} 
%%     &
%%     \infer[\ot]{\asem{S \mid \Gamma_1} \ot \asem{\n \mid \Gamma_2} 
%%             \tto C_1 \ot C_2}{
%%       \infer*[\sound\,f_1]{\asem{S \mid \Gamma_1} \tto C_1}{} 
%%       &
%%       \infer*[\sound\,f_2]{\asem{\n \mid \Gamma_2} \tto C_2}{}
%%     }
%%   }
%% }$
%% \end{tabbing}
% \sound\,\left(
%   \infer*[f_1]{S \mid\Gamma_1 \vdash C_1}{}
%   \right)
% \ot
% \sound\,\left(
%   \infer*[f_2]{\n \mid\Gamma_2 \vdash C_2}{}
%   \right)&
% 
%\]
\end{itemize}%
\end{proof}
% \sound \,({\ot}R {S} {Γ} {Δ} f g) = sound f ⊗ sound g ∘ lemma S Γ Δ\\

\section{An Equational Theory on Cut-Free Derivations}
\label{sec:adequacy}

In this section, we establish a bijective correspondence between the skew monoidal categorical calculus (the free skew monoidal category $\Fsk(\Var)$) and the skew monoidal sequent calculus. We show that the translations $\sound$ and $\cmplt$ witnessing soundness and completeness are mutually inverse up to an appropriate equational theory on derivations. 

%To state this correspondence, 
%we need to use the interpretation $\asem{~}$ of sequent calculus antecedents as% formulae. 
%In fact, this is not (yet) literally true for a couple of reasons.

An equational theory is needed because many different sequent calculus derivations can be sent by the soundness translation to equivalent categorical calculus derivations (corresponding to the same map in $\Fsk(\Var)$).
For example, the different derivations
\begin{equation}\label{deriv:diff1}
\small 
\tag{$*$}
\vcenter{
\infer[\scut]{\stseq{\n}{X}{X}}{
  \infer[\otR]{\stseq{\n}{X}{\I \ot X}}{
    \infer[\IR]{\stseq{\n}{~}{\I}}{} 
    & 
    \infer[\uf]{\stseq{\n}{X}{X}}{
      \infer[\ax]{\stseq{X}{~}{X}}{
      }
    }
  } 
  & 
  \infer[\otL]{\stseq{\I \ot X}{~}{X}}{
    \infer[\IL]{\stseq{\I}{X}{X}}{
      \infer[\uf]{\stseq{\n}{X}{X}}{
        \infer[\ax]{\stseq{X}{~}{X}}{} 
        }
      }
    }
  }
}
\qquad
\vcenter{
\infer[\uf]{\stseq{\n}{X}{X}}{
  \infer[\ax]{\stseq{X}{~}{X}}{
  }
}
}
\end{equation}
are sent by the translation $\sound$ to the same derivation $\id \comp \lambda_X : \I \ot X \tto \I$.
Likewise, the two different derivations
\begin{equation}\label{deriv:diff2} 
\small
\tag{$\dagger$}
\vcenter{\infer[\IL]{\stseq{\I}{~}{\I}}{\infer[\IR]{\stseq{\n}{~}{\I}}{}}} \qquad
\vcenter{\infer[\ax]{\stseq{\I}{~}{\I}}{}}
\end{equation}
are both mapped to the same derivation $\id : \I \tto \I$.
Here we begin to address this overabundance of derivations in two steps:
\begin{enumerate}
\item We restrict to \emph{cut-free} derivations by relying on eliminability of  $\scut$ and $\ccut$ (i.e., their admissibility in the cut-free fragment), and identify a derivation with cuts with the corresponding cut-free derivation.
\item We impose a further equivalence relation on cut-free derivations.
\end{enumerate}
In Section~\ref{sec:focusing}, we will take a third step of identifying \emph{canonical representatives} for the induced equivalence classes of derivations, corresponding to a natural focused subsystem of the sequent calculus.
This will complete the picture and provide a powerful coherence theorem for skew monoidal categories.

\begin{lemma}[Eliminability of cuts]\label{lem:admits-cut}
Each of the rules $\scut$ and $\ccut$
%% \[
%% \infer[\scut]{S \mid \Gamma, \Delta \vdash C}{
%%   S \mid \Gamma \vdash A
%%   & 
%%   A \mid \Delta \vdash C
%% }
%% \qquad
%% \infer[\ccut]{S \mid \Delta_0, \Gamma, \Delta_1 \vdash C}{
%%   \n \mid \Gamma \vdash A
%%   & 
%%   S \mid \Delta_0, A, \Delta_1 \vdash C
%% }
%% \]
is eliminable, i.e., admissible in the cut-free fragment in the sense that, given cut-free derivations of its premises, there is a cut-free derivation of its conclusion.
\end{lemma}
We elide the proof of cut admissibility here, since it follows the same pattern as the proof of Lemma~\ref{lemma:foc-admits} in Section~\ref{sec:focusing} used to establish the stronger focusing completeness result.\footnote{The proof appears in the conference version \cite{UVZ:seqsmcMFPS18} where cut admissibility is stated as Prop.~3.3.} 
We moreover assert the following:

\begin{lemma}\label{lem:skeweqns}
  The cut-elimination algorithm from the proof of Lemma~\ref{lem:admits-cut} validates the equations in Figures~\ref{fig:skeweqns} and \ref{fig:skeweqns-ctd}, in the sense that they hold for $\scut$ and $\ccut$ as defined operations on cut-free derivations.
\end{lemma}

\begin{figure}  %% strange: [p] puts the figure on the next page, while [t] puts it at the end of the document
\scriptsize
Unitality of identity wrt.\ cut 
\begin{align}
%\tag{Unitality of identity wrt.\ cut} \\
  \infer[\scut]{\tight{A}{\Delta}{C}}{
    \infer[\id]{\tight A~A}{} 
    & 
    \tight[f]{A}{\Delta}{C}
  }
  \quad &= \quad
  \tight[f]{A}{\Delta}{C} 
\label{eq:skewmulti1a}
  \\[1em]
  \infer[\ccut]{\flexi{S}{\Gamma,A,\Delta}{C}}{\infer[\uf]{\loose{A}{A}}{
    \infer[\id]{\tight A~A}{}} 
    & 
    \flexi[f]{S}{\Gamma,A,\Delta}{C}
  }
  \quad &= \quad 
  \flexi[f]{S}{\Gamma,A,\Delta}{C} 
\label{eq:skewmulti1b}
  \\[1em]
  \infer[\scut]{\flexi{S}{\Gamma}{A}}{\flexi[g]{S}{\Gamma}{A} & \infer[\id]{\tight A~A}{}}
   \quad &= \quad  
  \flexi[g]{S}{\Gamma}{A} 
\label{eq:skewmulti1c}
%\\[0.5em] \hline
%\tag{Associativity of cut}
%\\[0.5em]
\end{align}
Associativity of cut
%\tag{Associativity of cut}\\
\begin{multline}
\infer[\scut]{\flexi{S}{\Gamma,\Delta,\Lambda}{C}}{\infer[\scut]{\flexi{S}{\Gamma,\Delta}{B}}{\flexi[f]{S}{\Gamma}{A} & \tight[g]{A}{\Delta}{B}} & \hspace*{-3pt}\tight[h]{B}{\Lambda}{C}}  
\quad = \\
  \infer[\scut]{\flexi{S}{\Gamma,\Delta,\Lambda}{C}}{\flexi[f]{S}{\Gamma}{A} & 
    \hspace*{-5pt}\infer[\scut]{\tight{A}{\Delta,\Lambda}{C}}{\tight[g]{A}{\Delta}{B} & \tight[h]{B}{\Lambda}{C}}} \label{eq:skewmulti2a}
\end{multline}
\begin{multline}
\infer[\scut]{\flexi{S}{\Delta_0,\Gamma,\Delta_1,\Lambda}{C}}{\infer[\ccut]{\flexi{S}{\Delta_0,\Gamma,\Delta_1}{B}}{\loose[f]{\Gamma}{A} & \flexi[g]{S}{\Delta_0,A,\Delta_1}{B}} & \tight[h]{B}{\Lambda}{C}}
   \\
\quad = \\
\infer[\ccut]{\flexi{S}{\Delta_0,\Gamma,\Delta_1,\Lambda}{C}}{\loose[f]{\Gamma}{A} & 
    \infer[\scut]{\flexi{S}{\Delta_0,A,\Delta_1,\Lambda}{C}}{\flexi[g]{S}{\Delta_0,A,\Delta_1}{B} & \tight[h]{B}{\Lambda}{C}}} \label{eq:skewmulti2b}
\end{multline}
\begin{multline}
\infer[\ccut]{\flexi{S}{\Lambda_0,\Delta_0,\Gamma,\Delta_1,\Lambda_1}{C}}{\infer[\ccut]{\loose{\Delta_0,\Gamma,\Delta_1}{B}}{\loose[f]{\Gamma}{A} & \loose[g]{\Delta_0,A,\Delta_1}{B}} & \flexi[h]{S}{\Lambda_0,B,\Lambda_1}{C}}
\quad =  \\[6pt]
\infer[\ccut]{\flexi{S}{\Lambda_0,\Delta_0,\Gamma,\Delta_1,\Lambda_1}{C}}{
  \loose[f]{\Gamma}{A} 
  & 
  \infer[\ccut]{\flexi{S}{\Lambda_0,\Delta_0,A,\Delta_1,\Lambda_1}{B}}{
    \loose[g]{\Delta_0,A,\Delta_1}{B} 
    & 
    \flexi[h]{S}{\Lambda_0,B,\Lambda_1}{C}
  }
}
\label{eq:skewmulti2c}
\end{multline}
\caption{Equations governing the cut rules}
\label{fig:skeweqns}
\end{figure}

\begin{figure}  %% strange: [p] puts the figure on the next page, while [t] puts it at the end of the document
\scriptsize
Parallel cuts commute
\begin{multline}
\infer[\scut]{\flexi{S}{\Gamma_1,\Delta_1,\Gamma_2,\Delta_2}{C}}{
    \flexi[f_1]{S}{\Gamma_1}{A} &
    \infer[\ccut]{\tight{A}{\Delta_1,\Gamma_2,\Delta_2}{C}}{
      \loose[f_2]{\Gamma_2}{B} &
      \tight[g]{A}{\Delta_1,B,\Delta_2}{C}}}
  \quad = \\
  \infer[\ccut]{\flexi{S}{\Gamma_1,\Delta_1,\Gamma_2,\Delta_2}{C}}{
    \loose[f_2]{\Gamma_2}{B} &
    \infer[\scut]{\flexi{S}{\Gamma_1,\Delta_1,B,\Delta_2}{C}}{
      \flexi[f_1]{S}{\Gamma_1}{A} &
      \tight[g]{A}{\Delta_1,B,\Delta_2}{C}}} \label{eq:skewmulti3a}
\end{multline}
\begin{multline}
\infer[\ccut]{\flexi{S}{\Delta_0,\Gamma_1,\Delta_1,\Gamma_2,\Delta_2}{C}}{
    \loose[f_1]{\Gamma_1}{A} &
    \infer[\ccut]{\flexi{S}{\Delta_0,A,\Delta_1,\Gamma_2,\Delta_2}{C}}{
      \loose[f_2]{\Gamma_2}{B} &
      \flexi[g]{S}{\Delta_0,A,\Delta_1,B,\Delta_2}{C}}}
  \quad = \\
  \infer[\ccut]{\flexi{S}{\Delta_0,\Gamma_1,\Delta_1,\Gamma_2,\Delta_2}{C}}{
    \loose[f_2]{\Gamma_2}{B} &
    \hspace*{-5pt}\infer[\ccut]{\flexi{S}{\Delta_0,\Gamma_1,\Delta_1,B,\Delta_2}{C}}{
      \loose[f_1]{\Gamma_1}{A} &
      \hspace*{-2pt}\flexi[g]{S}{\Delta_0,A,\Delta_1,B,\Delta_2}{C}}} \label{eq:skewmulti3b}
\end{multline}
Cut commutes with shift
\begin{align}
\infer[\scut]{\loose{A,\Gamma,\Delta}{C}}{
  \infer[\uf]{\loose{A,\Gamma}{B}}{
    \tight[f]{A}{\Gamma}{B}} &
    \tight[g]{B}{\Delta}{C}}
  \quad &= \quad
\infer[\uf]{\loose{A,\Gamma,\Delta}{C}}{
  \infer[\scut]{\tight{A}{\Gamma,\Delta}{C}}{
    \tight[f]{A}{\Gamma}{B} &
    \tight[g]{B}{\Delta}{C}}} \label{eq:skewmulti4a}
\\[0.5em]
\infer[\ccut]{\loose{A,\Delta_0,\Gamma,\Delta_1}{C}}{
  \loose[f]{\Gamma}{B} &
  \infer[\uf]{\loose{A,\Delta_0,B,\Delta_1}{C}}{
    \tight[g]{A}{\Delta_0,B,\Delta_1}{C}}}
  \quad &= \quad
\infer[\uf]{\loose{A,\Delta_0,\Gamma,\Delta_1}{C}}{
  \infer[\ccut]{\tight{A}{\Delta_0,\Gamma,\Delta_1}{C}}{
    \loose[f]{\Gamma}{B} &
    \tight[g]{A}{\Delta_0,B,\Delta_1}{C}}} \label{eq:skewmulti4b}
\\[0.5em]
\infer[\ccut]{\loose{\Gamma,\Delta}{C}}{
  \loose[f]{\Gamma}{A} &
  \infer[\uf]{\loose{A,\Delta}{C}}{
    \tight[g]{A}{\Delta}{C}}}
  \quad &= \quad 
\infer[\scut]{\loose{\Gamma,\Delta}{C}}{
    \loose[f]{\Gamma}{A} &
    \tight[g]{A}{\Delta}{C}} \label{eq:skewmulti4c}
\end{align}
\caption{Equations governing the cut rules (continued)}
\label{fig:skeweqns-ctd}
\end{figure}

As shown by the example of \eqref{deriv:diff2}, however, restricting the domain of $\sound$ to cut-free derivations is not enough to obtain injectivity.
We therefore identify cut-free derivations up to an equivalence relation $\circeq$, defined as the least congruence induced by the equations in Figure~\ref{fig:cutfreeeqns}.

\begin{figure}[t]
\[
\footnotesize
\begin{array}{c@{\quad \circeq\quad}c}
\proofbox{
\infer[\ax]{\stseq{\I}{~}{\I}}{
}
}
&
\proofbox{
\infer[\IL]{\stseq{\I}{~}{\I}}{
  \infer[\IR]{\stseq{\n}{~}{\I}}{
  }
}
}
\\[6pt]
\proofbox{
\infer[\ax]{\stseq{A \ot B}{~}{A \ot B}}{
}
}
&
\proofbox{
\infer[\otL]{\stseq{A \ot B}{~}{A \ot B}}{
  \infer[\otR]{\stseq{A}{B}{A \ot B}}{
   \infer[\ax]{\stseq{A}{~}{A}}{
   }
   &
   \infer[\uf]{\stseq{\n}{B}{B}}{
     \infer[\ax]{\stseq{B}{~}{B}}{
     }
   }
 }
} 
}
\\[6pt]
\proofbox{
\infer[\otR]{\stseq{\n}{A', \Gamma, \Delta}{A \otimes B}}{
  \infer[\uf]{\stseq{\n}{A', \Gamma}{A}}{
    \stseq{A'}{\Gamma}{A}
  }
  &
  \stseq{\n}{\Delta}{B}
}
}
&
\proofbox{
\infer[\uf]{\stseq{\n}{A', \Gamma, \Delta}{A \otimes B}}{
  \infer[\otR]{\stseq{A'}{\Gamma, \Delta}{A \otimes B}}{
    \stseq{A'}{\Gamma}{A}
    &
    \stseq{\n}{\Delta}{B}
  }
}
}
\\[6pt]
\proofbox{
\infer[\otR]{\stseq{\I}{\Gamma, \Delta}{A \otimes B}}{
  \infer[\IL]{\stseq{\I}{\Gamma}{A}}{
    \stseq{\n}{\Gamma}{A}
  }
  &
  \stseq{\n}{\Delta}{B}
}
}
&
\proofbox{
\infer[\IL]{\stseq{\I}{\Gamma, \Delta}{A \otimes B}}{
  \infer[\otR]{\stseq{\n}{\Gamma, \Delta}{A \otimes B}}{
    \stseq{\n}{\Gamma}{A}
    &
    \stseq{\n}{\Delta}{B}
  }
}
}
\\[6pt]
\proofbox{
\infer[\otR]{\stseq{A' \ot B'}{\Gamma, \Delta}{A \otimes B}}{
  \infer[\otL]{\stseq{A' \ot B'}{\Gamma}{A}}{
    \stseq{A'}{B', \Gamma}{A}
  }
  &
  \stseq{\n}{\Delta}{B}
}
}
&
\proofbox{
\infer[\otL]{\stseq{A' \ot B'}{\Gamma, \Delta}{A \otimes B}}{
  \infer[\otR]{\stseq{A'}{B', \Gamma, \Delta}{A \otimes B}}{
    \stseq{A'}{B', \Gamma}{A}
    &
    \stseq{\n}{\Delta}{B}
  }
}
}
\end{array}
\]
\caption{Equivalences imposed on cut-free derivations}
\label{fig:cutfreeeqns}
\end{figure}
%We now need to show that the operations $\sound$ and $\cmplt$ are compatible with both equivalence relations $\doteq$ and $\circeq$.
\begin{lemma}\label{lem:soundcompat}
For all $f , g : \stseq{S}{\Gamma}{C}$, $f \circeq g$ implies
$\sound \, f \doteq \sound\, g$. 
\end{lemma}
\begin{proof}
  By induction on the proof of $f \circeq g$.
\end{proof}
% nz: commenting this out as a low-level detail available in the Agda code
%% \begin{lemma}
%% \begin{enumerate}
%% \item For all $f_1, f_2 : S \mid \Gamma \vdash A$ and $g_1,g_2 : A \mid \Delta
%% \vdash C$, if $f_1 \circeq f_2$ and $g_1 \circeq g_2$, then $\scut
%% \,f_1\,g_1 \circeq \scut\,f_2\,g_2$.
%% \item For all $f_1, f_2 : \n \mid \Gamma \vdash A$ and $g_1,g_2 : S \mid
%%   \Delta_0, A , \Delta_1 \vdash C$, if $f_1 \circeq f_2$ and $g_1 \circeq g_2$, then $\ccut
%% \,f_1\,g_1 \circeq \ccut\,f_2\,g_2$.
%% \end{enumerate}
%% \end{lemma}

\begin{lemma}\label{lem:cmpltcompat}
For all $f , g : A \tto C$, $f \doteq g$ implies $\cmplt\, f \circeq \cmplt\, g$.
\end{lemma}
\begin{proof}
  By induction on the proof of $f \doteq g$, after showing that the defined operations $\scut$ and $\ccut$ are compatible with the relation $\circeq$, and relying on the fact that they satisfy the equations in Figures \ref{fig:skeweqns} and \ref{fig:skeweqns-ctd}.
\end{proof}
It is possible to show that soundness is left inverse to completeness.
\begin{lemma}\label{lem:soundcmplt}
For all $f : A \tto C$, we have $\sound\,(\cmplt\,f) \doteq f$.
\end{lemma}
\begin{proof}
  By induction on $f$.
  We only show the proof of the
  case $f = f_1 \ot f_2$, where $f_1 : A \tto C$ and $f_2 : B \tto D$.
\small
\begin{align*}
%\begin{array}{r@{~}c@{~}l}
\sound\,\left( \cmplt\, (f_1 \ot f_2)\right)
& \doteq
\sound \, \left(
\proofbox{
\infer[\otL]{\stseq{A \ot B}{~}{C \ot D}}{
\infer[\otR]{\stseq{A}{B}{C \ot D}}{
  \stseq[\cmplt\, f_1]{A}{~}{C}
  & 
  \infer[\uf]{\stseq{\n}{B}{D}}{
    \stseq[\cmplt\, f_2]{B}{~}{D}
  }
}}
}
\right) \tag{defn. $\cmplt$} 
\\
& \doteq 
(\sound\,(\cmplt\,f_1)\ot(\sound\,(\cmplt\,f_2) \circ \lam)\circ (\al \comp (\rho \ot \id)) \tag{defn. $\sound$}\\
& \doteq
(f_1\ot(f_2 \circ \lam))\circ (\al \comp (\rho \ot \id)) \tag{induction hypothesis} \\
%& = (f_1\ot f_2) \circ ((\id \ot \lam) \comp \al \comp (\rho \ot \id)) \tag{skew monoidal equations} \\
& \doteq f_1 \ot f_2 \tag{skew monoidal equations}
\\ &\tag*{$\square$}  % hack to get qed symbol in right place
%  \end{array}
\end{align*}
%(We refer the reader to the Agda formalization for proofs of the other cases.)
%% \[
%% \small
%% \begin{array}{r@{~}c@{~}l}
%% \sound\,\left( \cmplt\, (h \comp g)\right)
%% & \eqdf &
%% \sound \, \left(
%% \proofbox{
%% \infer[\scut]{A \mid ~ \vdash C}{
%%   \infer*[\cmplt\, g]{A \mid ~ \vdash B}{}
%%   & 
%%   \infer*[\cmplt\, h]{B \mid ~ \vdash C}{}
%% }
%% }
%% \right) 
%% \\
%% & \stackrel{(1)}{\doteq} &
%% \proofbox{
%% \infer[\scomp]{A \tto C}{
%%   \infer*[\sound\,(\cmplt\,g)]{A \tto B}{}
%%   & 
%%   \infer*[\sound\,(\cmplt\,h)]{B \tto C}{}
%% }
%% } \\
%% & \eqdf &
%% \proofbox{
%% \infer[\dcomp]{A \tto C}{
%%   \infer[\dcomp]{A \tto B}{
%%     \infer[\id]{A \tto A}{}
%%     &
%%     \infer*[\sound\,(\cmplt\,g)]{A \tto B}{}
%%   }
%%   & 
%%   \infer*[\sound\,(\cmplt\,h)]{B \tto C}{}
%% }
%% } \\
%% & \doteq &
%% \proofbox{
%% \infer[\dcomp]{A \tto C}{
%%   \infer*[\sound\,(\cmplt\,g)]{A \tto B}{}
%%   & 
%%   \infer*[\sound\,(\cmplt\,h)]{B \tto C}{}
%% }
%% } \\
%% & \stackrel{(2)}{\doteq} &
%% \proofbox{
%% \infer[\dcomp]{A \tto C}{
%%   \infer*[g]{A \tto B}{}
%%   & 
%%   \infer*[h]{B \tto C}{}
%% }
%% }
%% \end{array}
%% \]
%% Equality $(1)$ follows from Lemma \ref{lem:soundcut}.
%% %Equality $(1.5)$
%% %follows from one of the generating equations of the relation
%% %$\doteq$ given in Section \ref{sec:skewmoncat}. 
%% Equality $(2)$ follows
%% from the induction hypothesis applied to $h$ and $g$.
\end{proof}
On the other hand, the postcomposition of completeness with soundness sends a derivation $f : \stseq{S}{\Gamma}{C}$ to a derivation of $\cmplt\,(\sound\,f) : \stseq{\asem{S \mid \Gamma}}{~}{C}$, and so for general $S$ and $\Gamma$, the two derivations are not directly comparable.
We can repair this discrepancy by realizing that there is a one-to-one correspondence between derivations $\stseq{S}{\Gamma}{C}$ and derivations $\stseq{\asem{S \mid \Gamma}}{~}{C}$, when considered up to $\circeq$.
\begin{lemma}[Invertibility of $\IL$ and $\otL$]
  \label{lem:invILotL}
  The following rules are admissible.
\[
\begin{array}{c@{\quad\quad}c}
\infer[\ILinv]{\stseq{\n}{\Gamma}{C}}{
  \stseq{\I}{\Gamma}{C}
}
&
\infer[\otLinv]{\stseq{A}{B,\Gamma}{C}}{
  \stseq{A \ot B}{\Gamma}{C}
}
\end{array}
\]
Moreover, $\ILinv$ and $\otLinv$ are compatible with $\circeq$, and inverse to the rules $\IL$ and $\otL$ in the sense that
\[ \ILinv\, (\IL\, f) = f \qquad \IL\, (\ILinv\, f) \circeq  \otLinv \]
\[ \otLinv\, (\otL\, f) = f \qquad \otL\, (\otLinv\, f) \circeq f \]
for all derivations $f$ of the appropriate type.
%% \[
%% \small
%% \proofbox{
%% \infer[\ILinv]{\stseq{\n}{\Gamma}{C}}{
%%   \infer[\IL]{\stseq{\I}{\Gamma}{C}}{
%%     \stseq[f]{\n}{\Gamma}{C}
%%   }
%% }
%% }
%% ~ = ~
%% \proofbox{
%% \stseq[f]{\n}{\Gamma}{C}
%% }
%% \hspace*{1cm}
%% \proofbox{
%% \infer[\otLinv]{\stseq{A}{B, \Gamma}{C}}{
%%   \infer[\otL]{\stseq{A \ot B}{\Gamma}{C}}{
%%     \stseq[f]{A}{B, \Gamma}{C}
%%   }
%% }
%% }
%% ~ = ~
%% \proofbox{
%% \stseq[f]{A}{B, \Gamma}{C}
%% }
%% \]
%% \[
%% \small
%% \proofbox{
%% \infer[\IL]{\I \mid \Gamma \vdash C}{
%%   \infer[\ILinv]{\n \mid \Gamma \vdash C}{
%%     \infer*[f]{\I \mid \Gamma \vdash C}{
%%     }
%%   }
%% }
%% }
%% ~ \circeq ~
%% \proofbox{
%% \infer*[f]{\I \mid \Gamma \vdash C}{
%% }
%% }
%% \hspace*{1cm}
%% \proofbox{
%% \infer[\otL]{A \ot B \mid \Gamma \vdash C}{
%%   \infer[\otLinv]{A \mid B, \Gamma \vdash C}{
%%     \infer*[f]{A \ot B \mid \Gamma \vdash C}{
%%     }
%%   }
%% }
%% }
%% ~ \circeq ~
%% \proofbox{
%% \infer*[f]{A \ot B \mid \Gamma \vdash C}{
%% }
%% }
%% \]
\end{lemma}
\begin{proof}
  By induction on derivations.
\end{proof}

%% \begin{lemma}\label{lem:inv}
%%   the following equations hold:
%% %  ${\I}L^{-1}\, ({\I}L\, f) =  f$ and ${\ot}L^{-1}\, ({\ot}L\, f) = f$.
%% \[
%% \small
%% \proofbox{
%% \infer[\ILinv]{\stseq{\n}{\Gamma}{C}}{
%%   \infer[\IL]{\stseq{\I}{\Gamma}{C}}{
%%     \stseq[f]{\n}{\Gamma}{C}
%%   }
%% }
%% }
%% ~ = ~
%% \proofbox{
%% \stseq[f]{\n}{\Gamma}{C}
%% }
%% \hspace*{1cm}
%% \proofbox{
%% \infer[\otLinv]{\stseq{A}{B, \Gamma}{C}}{
%%   \infer[\otL]{\stseq{A \ot B}{\Gamma}{C}}{
%%     \stseq[f]{A}{B, \Gamma}{C}
%%   }
%% }
%% }
%% ~ = ~
%% \proofbox{
%% \stseq[f]{A}{B, \Gamma}{C}
%% }
%% \]
%% \end{lemma}
%% \begin{lemma}\label{lem:leftinvcompat}
%% \begin{enumerate}
%% \item For all $f, g : \I \mid \Gamma \vdash C$, 
%%   if $f \circeq g$, then $\ILinv\, f \circeq \ILinv\, g$.
%% \item For all $f, g : A \ot B \mid \Gamma \vdash C$, 
%%   if $f \circeq g$, then $\otLinv\, f \circeq \otLinv\, g$.
%% \end{enumerate}
%% \noindent

\begin{lemma}\label{lem:leftinvstar}
For all $S$ and $\Gamma$, the following rule is admissible:
\[
\infer[\otLinvstar]{\stseq{S}{\Gamma , \Delta}{C}}{\stseq{\asem{S \mid \Gamma}}{\Delta}{C}}
\]
Moreover, $\otLinvstar$ is compatible with $\circeq$, and inverse to the rule
\[
\infer[\otLstar]{\stseq{\asem{S \mid \Gamma}}{\Delta}{C}}{\stseq{S}{\Gamma , \Delta}{C}}
\]
obtained by iterated application of $\otL$ and $\IL$, in the sense that

\[ 
\otLinvstar\, (\otLstar\, f) = f \qquad \otLstar\, (\otLinvstar\, f) \circeq f \]
\end{lemma}
\begin{proof}
  By induction on $\Gamma$, using Lemma~\ref{lem:invILotL}.
\end{proof}
Moreover, we can observe that the operation $\otLinvstar$ has no effect upon the soundness translation.
\begin{lemma}\label{lem:soundinv}
For all $f : \stseq{\asem{S\mid\Gamma}}{\Delta}{C}$, we have $\sound\,(\otLinvstar\,f) \doteq \sound\,f$.
%% \[
%% \small
%% \begin{array}{c}
%% \sound \left(
%% \proofbox{
%% \infer[\otLinvstar]{S \mid \Gamma \vdash C}{
%%   \infer*[f]{\asem{S \mid \Gamma} \mid ~ \vdash C}{}
%% }
%% }
%% \right)
%% ~ \doteq ~
%% \sound \left(
%% \proofbox{
%%   \infer*[f]{\asem{S \mid \Gamma} \mid ~ \vdash C}{}
%% }
%% \right)
%% \end{array}
%% \]
\end{lemma}
Finally, combining these results we can obtain our desired bijection.
\begin{lemma}[Strong completeness]\label{lem:strcompl} For any derivation 
  $f : \asem{S \mid \Gamma} \tto C$, there is a derivation
  $\strcmplt\,f : \stseq{S}{\Gamma}{C}$ given by
  $\strcmplt\,f \eqdf \otLinvstar\, (\cmplt\, f)$.
\end{lemma}

\begin{lemma}\label{lem:strcmpltcompat}
For all $f , g : \asem{S \mid \Gamma} \tto C$, $f \doteq g$ implies
$\strcmplt\, f \circeq \strcmplt\, g$.
\end{lemma}
\begin{lemma}\label{lem:soundstrcmplt}
For all $f : \asem{S \mid \Gamma} \tto C$, we have $\sound\,(\strcmplt\,f) \doteq f$.
\end{lemma}
\begin{proof}
Immediate from Lemmata~\ref{lem:soundcmplt} and \ref{lem:soundinv}.
\end{proof}

\begin{lemma}\label{lem:strcmpltsound}
 For all $f : \stseq{S}{\Gamma}{C}$, we have $\strcmplt \,(\sound\,f) \circeq f$.
\end{lemma}
\begin{proof}
  By induction on $f$.
%\nz{Maybe include a case or two of the proof?}
\end{proof}
\begin{theorem}
  The translations $\sound$ and $\strcmplt$ witness a one-to-one correspondence between derivations of $\stseq{S}{\Gamma}{C}$ (considered up to $\circeq$) and derivations of $\asem{S\mid\Gamma} \tto C$ (considered up to $\doteq$).
  As a special case, $\sound$ and $\cmplt$ witness a one-to-one correspondence between derivations $\stseq{A}{~}{C}$ and derivations $A \tto C$.
\end{theorem}
%% As a special case, the restriction of $\sound$ to sequents 
%% $\stseq{A}{~}{C}$ is a right inverse of $\cmplt$ up to $\circeq$.

%% \begin{corollary}\label{cor:cmpltsound}
%%  For all $f : A \mid ~ \vdash C$, we have $\cmplt \,(\sound\,f) \circeq
%%  f$.
%% \end{corollary}

%% Theorem \ref{thm:strcmpltsound} and Corollary~\ref{cor:soundstrcmplt}
%% show that $\sound$ and $\strcmplt$ make a bijection between the
%% derivations of $S \mid \Gamma \vdash C$ (considered up to $\circeq$) and the derivations of
%% $\asem{S \mid \Gamma} \tto C$ (considered up to $\doteq$). (And Corollary \ref{cor:cmpltsound} and
%% Theorem~\ref{thm:soundcmplt} demonstrate that restricted $\sound$ and $\cmplt$
%% form a bijection between the derivations of $A \mid ~ \vdash C$
%% and the derivations of $A \tto C$.)

\section{A Focused Subsystem of Canonical Derivations}
\label{sec:focusing}

% In the development of the previous section, we saw that we had too
% many derivations in the sequent calculus and we had to identify some
% by introducing the congruence $\circeq$. An alternative approach to
% cutting down the number of derivations is to identify a canonical
% representative of each equivalence class. With our sequent calculus,
% this is possible and easy. Indeed, 
If we consider the congruence relation $\circeq$ on cut-free 
derivations as a term rewrite system just by directing every equation
in Figure~\ref{fig:cutfreeeqns} from left to right, we can notice it is weakly
confluent and strongly normalizing, hence strongly confluent with
unique normal forms.  It turns out that these normal forms admit a simple
direct description, corresponding to a natural \emph{focused} subsystem
of the skew monoidal sequent calculus.

In the style of Andreoli \cite{Andreoli92}, we present the focused subsystem as a sequent calculus with an
additional \emph{phase annotation} on sequents (see Figure~\ref{fig:focusing}), which alternates between
$\mathsf{L}$ (``left phase'') and $\mathsf{R}$ (``right phase'').

\begin{figure}[t]
\[
\begin{array}{c@{\quad \quad}c@{\quad \quad}c}
\infer[\uf]{\stseqL{\n}{A, \Gamma}{C}}{
  \stseqL{A}{\Gamma}{C}
}
&
\infer[\switch]{\stseqL{T}{\Gamma}{C}}{
  \stseqR{T}{\Gamma}{C}
}
\hspace*{-12mm}
&
\infer[\axatm]{\stseqR{X}{~}{X}}{
}
\\[6pt]
\infer[\IL]{\stseqL{\I}{\Gamma}{C}}{
  \stseqL{\n}{\Gamma}{C}
}
& &
\infer[\IRfoc]{\stseqR{\n}{~}{\I}}{
}
\\[6pt]
\infer[\otL]{\stseqL{A \ot B}{\Gamma}{C}}{
  \stseqL{A}{B, \Gamma}{C}
}
& &
\hspace*{5mm}
\infer[\otRfoc]{\stseqR{T}{\Gamma, \Delta}{A \otimes B}}{
  \stseqR{T}{\Gamma}{A}
  &
  \stseqL{\n}{\Delta}{B}
} 
\end{array}
\]
Note: $T$ ranges over \emph{irreducible} stoups $T = X$ or $T = \n$.
\caption{Rules of the focused sequent calculus}
\label{fig:focusing}
\end{figure}

Observe that in an $\mathsf{L}$-sequent there is no restriction on the stoup, but
in an $\mathsf{R}$-sequent the stoup is forced to be \emph{irreducible,} meaning that it is either empty or atomic.\footnote{It is worth mentioning that Zeilberger \cite{Zei} did not need an explicit phase annotation in the focused sequent calculus for the Tamari order. This was because restricting to sequents with a non-empty stoup (in the absence of the $\uf$ rule and the rules for $\I$), the phase can be uniquely determined by asking whether the stoup is irreducible.}

The focused calculus is manifestly sound, in the sense that if one erases phase annotations, all of the rules are either rules of the original calculus or else (in the case of $\switch$) have conclusion equal to their premise.
\begin{proposition}[Focusing soundness]
For any focused derivation
$f : \stseqLR{S}{\Gamma}{C}{P}$ (where $P \in \{\mathsf{L}, \mathsf{R}\}$), there is a derivation $\emb_P\, f : \stseq{S}{\Gamma}{C}$.
\end{proposition}
Less obviously, the focused calculus is also complete, and indeed optimal in a sense that we will soon make precise.

As in Andreoli's original formulation for linear logic, we can think of the focused calculus as defining a backwards proof search strategy which attempts to build a derivation of a sequent, starting from the root.
Beginning in an $\mathsf{L}$-phase, the invertible rules $\IL$ and $\otL$ are applied to break down the formula in the stoup and transform it into a list of additional subformulae in the context.
Once the stoup is irreducible, there is a choice to either apply the $\uf$ rule (if the stoup is empty) to move another formula into the stoup and repeat the inversion process, or else apply the $\switch$ rule and go into $\mathsf{R}$-phase.
During an $\mathsf{R}$-phase, the non-invertible rule $\otR$ can be tried as necessary to attempt to continue proof search (which involves a non-deterministic splitting of the context, and moves back into $\mathsf{L}$-phase for the right premise), as can the rules $\IR$ and $\ax$ to attempt to finish off the derivation.
The crucial point is that this search strategy will always succeed in finding a proof, if one exists.

% In the converse direction, even without appealing to insights that the
% term rewriting perspective provides, we can define a function for
% focusing any given unrestricted derivation.

\begin{theorem}[Focusing completeness]\label{thm:foccmplt}
  For any derivation $f : \stseq{S}{\Gamma}{C}$, there is a
  focused derivation $\focus\, f : \stseqL{S}{\Gamma}{C}$.
\end{theorem}
The proof of focusing completeness follows a standard pattern: one first shows that all of the rules of the original calculus are admissible for $\mathsf{L}$-phase sequents, and then obtains Theorem~\ref{thm:foccmplt} as an immediate corollary by induction on $f$.
\begin{lemma}\label{lemma:foc-admits}
Each of the rules
\[
\small
\begin{array}{c}
\infer[\ax]{\stseqL{A}{~}{A}}{}
\qquad
\infer[\IR]{\stseqL{\n}{~}{\I}}{}
\qquad
\infer[\otR]{\stseqL{S}{\Gamma,\Delta}{A \ot B}}{\stseqL{S}{\Gamma}{A} & \stseqL{\n}{\Delta}{B}}
\\[2ex]
\infer[\scut]{\stseqL{S}{\Gamma, \Delta}{C}}{
  \stseqL{S}{\Gamma}{A}
  & 
  \stseqL{A}{\Delta}{C}
}
\qquad
\infer[\ccut]{\stseqL{S}{\Delta_0, \Gamma, \Delta_1}{C}}{
  \stseqL{\n}{\Gamma}{A}
  & 
  \stseqL{S}{\Delta_0, A, \Delta_1}{C}
}
\end{array}
\]
is admissible in the sense that given focused derivations of its premises there is a focused derivation of its conclusion.
\end{lemma}
\begin{proof}
  Rule $\IR$ is immediately derivable as $\IR = \switch(\IRfoc)$, while $\otR$ is admissible by induction on the derivation of the first premise.
  Admissibility of $\ax$ can be shown in various ways, % (cf.~\cite[Lemma~1.13]{Zei})
  perhaps most simply by induction on $A$ after first showing admissibility of $\otR$.
  For example, if $A = A_1 \ot A_2$, we derive the axiom $\stseqL{A_1\ot A_2}{~}{A_1\ot A_2}$ by invoking $\otR$ and by appeal to the induction hypothesis on $A_1$ and $A_2$:

\[
\small
  \infer[\otL]{\stseqL{A_1 \ot A_2}{~}{A_1 \ot A_2}}{
  \infer[\otR]{\stseqL{A_1}{A_2}{A_1 \ot A_2}}{
   \infer[\ax]{\stseqL{A_1}{~}{A_1}}{
   }
   &
   \infer[\uf]{\stseqL{\n}{A_2}{A_2}}{
     \infer[\ax]{\stseqL{A_2}{~}{A_2}}{
     }
   }
 }
}
  \]
  Admissibility of the cut rules is only a bit more elaborate, requiring us to first introduce two more cut rules involving $\mathsf{R}$-phase sequents,

\[
\small
\infer[\scutR]{\stseqL{T}{\Gamma, \Delta}{C}}{
  \stseqR{T}{\Gamma}{A}
  & 
  \stseqL{A}{\Delta}{C}
}
\quad
\infer[\ccutR]{\stseqR{T}{\Delta_0, \Gamma, \Delta_1}{C}}{
  \stseqL{\n}{\Gamma}{A}
  & 
  \stseqR{T}{\Delta_0, A, \Delta_1}{C}
}
\]
  and then prove admissibility of the four rules by a mutual lexicographic induction on the cut formula $A$ and on the pair of derivations of the premises.
  The proof is relatively long but mechanical. We show all essential cases below.
  \begin{description}
  \item~ Definition of $\scut(f,g)$ for $f : \stseqL{S}{\Gamma}{A}$ and $g : \stseqL{A}{\Delta}{C}$:
    \begin{itemize}
    \item Case $f = \uf\,f'$ for some $f' : \stseqL{A'}{\Gamma'}{A}$. In particular, $S = \n$ and $\Gamma = A',\Gamma'$. We define:
\begin{multline*}
\small
\proofbox{
\infer[\scut]{\stseqL{\n}{A',\Gamma',\Delta}{C}}{
\infer[\uf]{\stseqL{\n}{A', \Gamma'}{A}}{
  \stseqL[f']{A'}{\Gamma'}{A}
}
  & 
  \stseqL[g]{A}{\Delta}{C}
}
}
\quad \eqdf \\[-10pt]
\small
\proofbox{
\infer[\uf]{\stseqL{\n}{A',\Gamma',\Delta}{C}}{
  \infer[\scut]{\stseqL{A'}{\Gamma',\Delta}{C}}{
    \stseqL[f']{A'}{\Gamma'}{A} &
    \stseqL[g]{A}{\Delta}{C}
  }
}
}
\end{multline*}
    \item Case $f = \IL\,f'$ or $\otL\,f'$. Similar to $f = \uf\,f'$.
    \item Case $f = \switch\,f'$ for some $f' : \stseqR{T}{\Gamma}{A}$. In particular, $S = T$ is irreducible. We define:
\[
\small
\proofbox{
\infer[\scut]{\stseqL{T}{\Gamma,\Delta}{C}}{
\infer[\switch]{\stseqL{T}{\Gamma}{C}}{
  \stseqR[f']{T}{\Gamma}{A}
}
  & 
  \stseqL[g]{A}{\Delta}{C}
}
}
~ \eqdf ~
\proofbox{
\infer[\scutR]{\stseqL{T}{\Gamma,\Delta}{C}}{
    \stseqR[f']{T}{\Gamma}{A} &
    \stseqL[g]{A}{\Delta}{C}
}
}
\]
    \end{itemize}
  \item~ Definition of $\scutR(f,g)$ for $f : \stseqR{T}{\Gamma}{A}$ and $g : \stseqL{A}{\Delta}{C}$:
    \begin{itemize}
    \item Case $f \eqdf \axatm$. In particular, $T = X = A$ and $\Gamma$ is
      empty. We define:

\[
\small
\proofbox{
\infer[\scutR]{\stseqL{X}{\Delta}{C}}{
  \infer[\axatm]{\stseqR{X}{~}{X}}{}
  & 
  \stseqL[g]{X}{\Delta}{C}
}
}
~ \eqdf ~
%\proofbox{
  \stseqL[g]{X}{\Delta}{C}
%}
\]
    \item Case $f = \IRfoc$. In particular, $T = \n$, $\Gamma$ is empty, and $A = \I$. Then $g$ is necessarily of the form $g = \IL\,g'$ for some $g' : \stseqL{\n}{\Delta}{C}$, and we define:
\[
\small
\proofbox{
\infer[\scutR]{\stseqL{\n}{\Delta}{C}}{
  \infer[\IRfoc]{\stseqR{\n}{~}{\I}}{}
  & 
  \infer[\IL]{\stseqL{\I}{\Delta}{C}}{\stseqL[g']{\n}{\Delta}{C}}
}
}
~ \eqdf ~
%\proofbox{
\stseqL[g']{\n}{\Delta}{C}
%}
\]

    \item Case $f = \otRfoc(f_1,f_2)$ for some $f_1 : \stseqR{T}{\Gamma_1}{A_1}$ and $f_2 : \stseqL{\n}{\Gamma_2}{A_2}$. In particular, $\Gamma = \Gamma_1,\Gamma_2$ and $A = A_1 \ot A_2$. Then $g$ is necessarily of the form $g = \otL\,g'$ for some $g' : \stseqL{A_1}{A_2,\Delta}{C}$, and we define:
\begin{multline*}
\small
\proofbox{
\infer[\scutR]{\stseqL{T}{\Gamma_1,\Gamma_2,\Delta}{C}}{
  \infer[\otRfoc]{\stseqR{T}{\Gamma_1,\Gamma_2}{A_1\ot A_2}}{\stseqR[f_1]{T}{\Gamma_1}{A_1} & \stseqL[f_2]{\n}{\Gamma_2}{A_2}}
  & 
  \infer[\otL]{\stseqL{A_1\ot A_2}{\Delta}{C}}{\stseqL[g']{A_1}{A_2,\Delta}{C}}
}
}
\quad \eqdf \\
\small
\proofbox{
\infer[\ccut]{\stseqL{T}{\Gamma_1,\Gamma_2,\Delta}{C}}{
  \stseqL[f_2]{\n}{\Gamma_2}{A_2} &
  \infer[\scutR]{\stseqL{T}{\Gamma_1,A_2,\Delta}{C}}{
    \stseqR[f_1]{T}{\Gamma_1}{A_1} &
    \stseqL[g']{A_1}{A_2,\Delta}{C}
  }
}
}
\end{multline*}
    \end{itemize}
  \item~ Definition of $\ccut(f,g)$ for $f : \stseqL{\n}{\Gamma}{A}$ and $g : \stseqL{S}{\Delta_0,A,\Delta_1}{C}$:
    \begin{itemize}
    \item Case $g = \uf\,g'$ for some $g' : \stseqL{A'}{\Delta'}{C}$. In particular, $S = \n$ and $\Delta_0,A,\Delta_1 = A',\Delta'$. There are two cases, depending on whether or not the context $\Delta_0$ is empty.
      \begin{itemize}
      \item If $\Delta_0$ is empty, then $A' = A$ and $\Delta' = \Delta_1$. We define:
\[
\small
      \proofbox{
        \infer[\ccut]{\stseqL{\n}{\Gamma,\Delta_1}{C}}{
          \stseqL[f]{\n}{\Gamma}{A} &
          \infer[\uf]{\stseqL{\n}{A , \Delta_1}{C}}{
            \stseqL[g']{A}{\Delta_1}{C}
            }
          }
      }
 ~ \eqdf ~ 
      \proofbox{
        \infer[\scut]{\stseqL{\n}{\Gamma,\Delta_1}{C}}{
          \stseqL[f]{\n}{\Gamma}{A} &
           \stseqL[g']{A}{\Delta_1}{C}
          }
      }    
\]

    \item If $\Delta_0 \eqdf A'',\Delta'_0$, then $A'' = A'$ and
      $\Delta' = \Delta'_0,A,\Delta_1$. We define:
\begin{multline*}
\small
      \proofbox{
        \infer[\ccut]{\stseqL{\n}{A',\Delta'_0,\Gamma,\Delta_1}{C}}{
          \stseqL[f]{\n}{\Gamma}{A} &
          \infer[\uf]{\stseqL{\n}{A' ,\Delta'_0,A, \Delta_1}{C}}{
            \stseqL[g']{A'}{\Delta'_0,A,\Delta_1}{C}
            }
          }
      }
 \quad \eqdf \\
\small
      \proofbox{
        \infer[\uf]{\stseqL{\n}{A',\Delta'_0,\Gamma,\Delta_1}{C}}{
          \infer[\ccut]{\stseqL{A'}{\Delta'_0,\Gamma,\Delta_1}{C}}{
            \stseqL[f]{\n}{\Gamma}{A} &
            \stseqL[g']{A'}{\Delta'_0,A,\Delta_1}{C}
          }
          }
      }    
\end{multline*} 
      \end{itemize}
    \item Case $g = \IL\,g'$ for some $g' : \stseqL{\n}{\Delta_0,A,\Delta_1}{C}$. In particular, $S = \I$. We define:
\begin{multline*}
\small
      \proofbox{
        \infer[\ccut]{\stseqL{\I}{\Delta_0,\Gamma,\Delta_1}{C}}{
          \stseqL[f]{\n}{\Gamma}{A} &
          \infer[\IL]{\stseqL{\I}{\Delta_0,A, \Delta_1}{C}}{
            \stseqL[g']{\n}{\Delta_0,A,\Delta_1}{C}
            }
          }
      }
\quad \eqdf \\[-10pt]
\small
      \proofbox{
        \infer[\IL]{\stseqL{\I}{\Delta_0,\Gamma,\Delta_1}{C}}{
          \infer[\ccut]{\stseqL{\n}{\Delta_0,\Gamma,\Delta_1}{C}}{
            \stseqL[f]{\n}{\Gamma}{A} &
            \stseqL[g']{\n}{\Delta_0,A,\Delta_1}{C}
          }
          }
      }    
\end{multline*}
    \item Case $g = \otL\,g'$: similar to $g = \IL\,g'$.
    \item Case $g = \switch\,g'$ for some $g' : \stseqR{T}{\Delta_0,A,\Delta_1}{C}$. In particular, $S = T$ is irreducible. We define:
\begin{multline*}
\small
      \proofbox{
        \infer[\ccut]{\stseqL{T}{\Delta_0,\Gamma,\Delta_1}{C}}{
          \stseqL[f]{\n}{\Gamma}{A} &
          \infer[\switch]{\stseqL{T}{\Delta_0,A, \Delta_1}{C}}{
            \stseqR[g']{T}{\Delta_0,A,\Delta_1}{C}
            }
          }
      }
\quad \eqdf \\
\small
      \proofbox{
        \infer[\switch]{\stseqL{T}{\Delta_0,\Gamma,\Delta_1}{C}}{
          \infer[\ccutR]{\stseqR{T}{\Delta_0,\Gamma,\Delta_1}{C}}{
            \stseqL[f]{\n}{\Gamma}{A} &
            \stseqR[g']{T}{\Delta_0,A,\Delta_1}{C}
          }
          }
      }    
\end{multline*}
    \end{itemize}
    
  \item Def.\ of $\ccutR(f,g)$ for $f : \stseqL{\n}{\Gamma}{A}$ and $g : \stseqR{T}{\Delta_0,A,\Delta_1}{C}$:
    \begin{itemize}
    \item Case $g = \otRfoc(g_1,g_2)$ for some $g_1 : \stseqR{T}{\Lambda_1}{C_1}$ and $g_2 : \stseqL{\n}{\Lambda_2}{C_2}$. In particular, $\Lambda_1,\Lambda_2 = \Delta_0,A,\Delta_1$ and $C = C_1 \ot C_2$. We proceed by checking if the formula $A$ occurs in $\Lambda_1$ or $\Lambda_2$.
      \begin{itemize}
      \item If $A$ occurs in $\Lambda_1$, we have $\Lambda_1 =
        \Delta_0,A,\Delta'_1$ and $\Delta_1 = \Delta'_1,\Lambda_2$ for some $\Delta'_1$. We define:
\begin{multline*}
\small
     \proofbox{
     \infer[\ccutR]{\stseqR{T}{\Delta_0,\Gamma,\Delta'_1,\Lambda_2}{C_1\ot C_2}}{
       \stseqL[f]{\n}{\Gamma}{A} &
       \infer[\otR]{\stseqR{T}{\Delta_0,A, \Delta'_1 ,\Lambda_2}{C_1 \ot C_2}}{
         \stseqR[g_1]{T}{\Delta_0,A,\Delta'_1}{C_1} &
         \stseqL[g_2]{\n}{\Lambda_2}{C_2}
       }
     }
     }
\quad \eqdf \\ 
\small
  \proofbox{
   \infer[\otR]{\stseqR{T}{\Delta_0,\Gamma,\Delta'_1,\Lambda_2}{C_1\ot C_2}}{
     \infer[\ccutR]{\stseqR{T}{\Delta_0,\Gamma,\Delta'_1}{C_1}}{
       \stseqL[f]{\n}{\Gamma}{A} &
       \stseqR[g_1]{T}{\Delta_0,A,\Delta'_1}{C_1}
     } &
     \stseqL[g_2]{\n}{\Lambda_2}{C_2}
     }
      }    
\end{multline*}

      \item If $A$ occurs in $\Lambda_2$, we have $\Lambda_2 =
        \Delta'_0,A,\Delta_1$ and $\Delta_0 = \Lambda_1,\Delta'_0$ for some $\Delta'_0$. We define:
\begin{multline*}
 \small
      \proofbox{
    \infer[\ccutR]{\stseqR{T}{\Lambda_1,\Delta'_0,\Gamma,\Delta_1}{C_1\ot C_2}}{
        \stseqL[f]{\n}{\Gamma}{A} 
       &
        \infer[\otR]{\stseqR{T}{\Lambda_1,\Delta'_0,A, \Delta_1}{C_1 \ot C_2}}{
          \stseqR[g_1]{T}{\Lambda_1}{C_1} &
          \stseqL[g_2]{\n}{\Delta'_0,A,\Delta_1}{C_2}
             }
         }
       }
\quad \eqdf \\ 
\small
  \proofbox{
   \infer[\otR]{\stseqR{T}{\Lambda_1,\Delta'_0,\Gamma,\Delta_1}{C_1\ot C_2}}{
     \stseqR[g_1]{T}{\Lambda_1}{C_1} &
     \infer[\ccut]{\stseqL{\n}{\Delta'_0,\Gamma,\Delta_1}{C_2}}{
       \stseqL[f]{\n}{\Gamma}{A} &
       \stseqL[g_2]{\n}{\Delta'_0,A,\Delta_1}{C_2}
     }
   }
   }
\end{multline*}
      \end{itemize}

    \item Case $g = \IRfoc$ or $\axatm$: impossible.
    \end{itemize}
  \end{description}
\end{proof}
We should mention that in our formalization, the domain of the $\focus$ function is restricted to cut-free derivations, factoring the proof of Theorem~\ref{thm:foccmplt} via Lemma~\ref{lem:admits-cut}.
We can then prove the following lemmata easily by induction, and derive the first half of our main coherence theorem by combining Corollary~\ref{cor:circeqifffoceq} below with the results of Sections~\ref{sec:seqcalc} and \ref{sec:adequacy}.
% The importance of the focusing completeness theorem comes from the fact that focused derivations provide \emph{canonical} representatives for their unfocused counterparts, as established by the following lemmata.

\begin{lemma}\label{lem:focemb}
  For any $f : \stseqL{S}{\Gamma}{C}$, $\focus\, (\embL\, f) = f$.
\end{lemma}
% \begin{proof}                    
% By induction on $f$.
% \end{proof}

\begin{lemma}\label{lem:foccirceq}
  For any $f, g : \stseq{S}{\Gamma}{C}$, if $f \circeq g$, then
  $\focus\, f = \focus\, g$.
\end{lemma}                       

%% \begin{proof}
%% By induction on the proof of $f \circeq g$.
%% \end{proof}

\begin{lemma}\label{lem:embfoc}
  For any $f : \stseq{S}{\Gamma}{C}$, we have
  $\embL\, (\focus\, f) \circeq f$.
\end{lemma}

%% \begin{proof}
%% By induction on $f$.
%% \end{proof}

\begin{corollary}\label{cor:circeqifffoceq}
  For any $f, g : \stseq{S}{\Gamma}{C}$, $f \circeq g$ iff $\focus\, f = \focus\, g$.
\end{corollary}

%% \begin{proof}
%%   By combining Lemmata~\ref{lem:focemb}, \ref{lem:foccirceq}, and \ref{lem:embfoc}.
%% \end{proof}
\noindent

%TARMO TRIES AGAIN: We should mention that in our formalization, the domain of the $\focus$ function is restricted to cut-free derivations, factoring the proof of Theorem~\ref{thm:foccmplt} via Lemma~\ref{lem:admits-cut}.
%This allows us to prove Lemmata~\ref{lem:foccirceq} and \ref{lem:embfoc} just for cut-free derivations and still conclude  Corollary~\ref{cor:circeqifffoceq}.

%CURRENT: We should mention that in our formalization, the domain of the $\focus$ function is restricted to cut-free derivations, factoring the proof of Theorem~\ref{thm:foccmplt} via Lemma~\ref{lem:admits-cut}.
%This allows us to prove Lemmata~\ref{lem:focemb}--\ref{lem:embfoc} easily by induction, and the first half of our main coherence theorem by combining Corollary~\ref{cor:circeqifffoceq} with the results of Sections~\ref{sec:seqcalc} and \ref{sec:adequacy}.

\begin{theorem}[Coherence: equality]\label{thm:coh:eq}
  For any $f, g : A \tto C$, we have $f \doteq g$ if and only if
  $\focus\, (\cmplt\, f) = \focus\, (\cmplt\, g)$.
\end{theorem}
%This coherence theorem may be compared with the one by Lack and Street \cite{LS:triosm}.
\noindent
Alternatively, it should be possible to prove Theorem~\ref{thm:coh:eq} more directly without the intermediate step via cut-free-but-unfocused derivations, by reproving analogues of some of the results of Section~\ref{sec:adequacy} directly for focused derivations.
However, we have not formalized this proof strategy.

We remark that Theorem~\ref{thm:coh:eq} gives a simple algorithm for deciding equality of maps in $\Fsk(\Var)$.
Moreover, as already mentioned, the rules of Figure~\ref{fig:focusing} can also be interpreted as defining a proof search strategy, and thus an algorithm for deciding existence of maps in $\Fsk(\Var)$.
Indeed, the rules can be turned into a simple algorithm for \emph{enumerating} all elements of any homset in the free skew monoidal category without duplicates, yielding the second half of our coherence theorem.
\begin{lemma}
  For any $S, \Gamma, C$, one can compute a finite list
  $\focderivs\, (S, \Gamma, C)$ of derivations of
  $\stseqL{S}{\Gamma}{C}$ containing every such derivation exactly once. In
  particular, we can decide whether $\stseqL{S}{\Gamma}{C}$ is
  derivable.
\end{lemma}

\begin{proof}
  As explained, we can consider the focused calculus as defining a
  root-first search strategy.  This search is guaranteed to terminate
  with a finite set of derivations because, for any goal sequent $\stseqLR{S}{\Gamma}{C}{P}$
  ($P \in \{\mathsf{L}, \mathsf{R}\}$), there are only finitely many
  possible instances of rules to apply, and the subgoals that they
  generate are always smaller relative to a well-founded order on
  sequents. (We can rank sequents by lexicographically ordered triples
  consisting of the number of occurrences of $\I$ and $\ot$,
  the information whether the stoup is empty or not, with singleton $<$ empty,
  and the phase, with $\mathsf{R} < \mathsf{L}$.)
  % We can rank sequents by lexicographically ordered
  % quadruples consisting of the number of times $\I$ and $\otimes$
  % occur in the succedent, the information whether the stoup is empty
  % or not, with non-empty $<$ empty, the number of times $\I$ and
  % $\otimes$ occur in the antecedent, and the mode, with
  % $\mathsf{R} < \mathsf{L}$.)
\end{proof}

% As a corollary, we obtain a function for enumerating, without
% duplicates, the maps of any homset of the free skew monoidal category,
% in particular, for deciding existence of a map.

\begin{theorem}[Coherence: enumeration]\label{thm:coh:enum}
  For any $A, C \in \Tm$, let
  \[ \fskmaps\,(A,C) \eqdf [ \sound\,(\embL\,f) \mid f \in \focderivs\, (A, (), C) ]. \]
  For any $f : A \tto C$, there exists a unique $g \in \fskmaps\,(A,C)$ such that $f \doteq g$.
  %contains exactly one representative from every $\doteq$-equivalence class of maps $A \tto B$ in $\Fsk(\Var)$.
\end{theorem}

% We note that focusing is not essential for defining a function
% enumerating all maps of a homset. Such a function is obtainable also
% from the unfocused yet cut-free sequent calculus of
% Section~\ref{sec:seqcalc}. The unfocused calculus is more
% non-deterministic than the focused calculus, but nonetheless only
% finitely many rule instances apply to any goal sequent and the
% subgoals are smaller with respect to a suitable well-founded
% order. But the straightforwardly defined enumeration function from the
% unfocused calculus delivers duplicates; additional work has to be put
% into detecting and removing them. The focused calculus avoids this
% problem.

\section{Comparison with Bourke and Lack}
\label{sec:bourke-lack}

As mentioned in the introduction, the analysis we have presented here
is closely related to Bourke and Lack's recent characterization of
skew monoidal categories as \emph{left representable skew multicategories}
\cite{BL:multi}.
In this section we describe the relationship more explicitly.

To set the stage, let us begin by recalling the concept of a \emph{multicategory} \cite{Lei}, but following the original sequent calculus-inspired formulation given by Lambek \cite{Lambek1969}.
Thus, an ordinary \defn{multicategory} $\M$ consists first of all of a set of objects, and for list of objects $A_1,\dots,A_n$ and any object $C$, a set $\M(A_1,\dots,A_n;C)$ of \emph{multimaps} with domain $A_1,\dots,A_n$ and codomain $C$.
We write $g : \multi{\Gamma}{C}$ or $\multi[g]{\Gamma}{C}$ to depict that $g$ is a multimap in $\M(\Gamma;C)$ where $\Gamma = A_1,\dots,A_n$.
Moreover, a multicategory must include, for every object $A$, an \emph{identity} multimap $\id_A : \multi{A}{A}$; and for every pair of multimaps $f : \multi{\Gamma}{A}$ and $g : \multi{\Delta_0,A,\Delta_1}{C}$, a \emph{composition} multimap $\cut_{\Delta_0-\Delta_1}(f,g) : \multi{\Delta_0,\Gamma,\Delta_1}{C}$.
We typically leave off the subscripts for $\id$ and $\cut$ when clear from context.
Finally, all of this data must be subject to four equations:

\begingroup
\allowdisplaybreaks
\scriptsize
\begin{align}
  \infer[\cut]{\multi{\Delta_0,A,\Delta_1}{C}}{\infer[\id]{\multi {A}{A}}{} & \multi[f]{\Delta_0,A,\Delta_1}{C}}
 \quad & =\quad 
\multi[f]{\Delta_0,A,\Delta_1}{C} \label{eq:multi1}
  \\[1em]
\infer[\cut]{\multi{\Gamma}{A}}{\multi[g]{\Gamma}{A} & \infer[\id]{\multi AA}{}}
\quad & =\quad  
\multi[g]{\Gamma}{A} \label{eq:multi2}
\end{align}
\begin{multline}
\infer[\cut]{\multi{\Lambda_0,\Delta_0,\Gamma,\Delta_0,\Lambda_1}{C}}{\infer[\cut]{\multi{\Delta_0,\Gamma,\Delta_1}{B}}{\multi[f]{\Gamma}{A} & \multi[g]{\Delta_0,A,\Delta_1}{B}} & \hspace*{-3pt}\multi[h]{\Lambda_0,B,\Lambda_1}{C}}
\quad \\
\infer[\cut]{\multi{\Lambda_0,\Delta_0,\Gamma,\Delta_1,\Lambda_1}{C}}{\multi[f]{\Gamma}{A} & 
    \hspace*{-5pt}\infer[\cut]{\multi{\Lambda_0,\Delta_0,A,\Delta_1,\Lambda_1}{B}}{\multi[g]{\Delta_0,A,\Delta_1}{B} & \multi[h]{\Lambda_0,B,\Lambda_1}{C}}} \label{eq:multi3}
\end{multline}
\begin{multline}
\infer[\cut]{\multi{\Delta_0,\Gamma_1,\Delta_1,\Gamma_2,\Delta_2}{C}}{
    \multi[f_1]{\Gamma_1}{A} &
    \infer[\cut]{\multi{\Delta_0,A,\Delta_1,\Gamma_2,\Delta_2}{C}}{
      \multi[f_2]{\Gamma_2}{B} &
      \multi[g]{\Delta_0,A,\Delta_1,B,\Delta_2}{C}}}
\quad = \\
  \infer[\cut]{\multi{\Delta_0,\Gamma_1,\Delta_1,\Gamma_2,\Delta_2}{C}}{
    \multi[f_2]{\Gamma_2}{B} &
    \infer[\cut]{\multi{\Delta_0,\Gamma_1,\Delta_1,B,\Delta_2}{C}}{
      \multi[f_1]{\Gamma_1}{A} &
      \multi[g]{\Delta_0,A,\Delta_1,B,\Delta_2}{C}}} \label{eq:multi4}
\end{multline}
\endgroup

Here, following Lambek, we have chosen to present the equations using a proof-theoretic notation whose meaning should hopefully be clear.
For example, equation \eqref{eq:multi1} can also be written $\cut_{\Delta_0-\Delta_1}(\id_A,f) = f$.
Likewise, we have taken as basic structure the operations
\[ 
\cut_{\Delta_0-\Delta_1} : \M(\Gamma;A) \times \M(\Delta_0,A,\Delta_1;C) \to \M(\Delta_0,\Gamma,\Delta_1;C) 
\]
which are sometimes referred to as ``partial'' composition (or substitution) operations, since they compose the first multimap into a single argument of the second multimap.
Multicategories may be alternatively defined (cf.~\cite{Lei}) using ``parallel'' composition operations of type
\[
  \M(\Gamma_1;A_1) \times \dots \M(\Gamma_n;A_n) \times \M(A_1,\dots,A_n;C) \to \M(\Gamma_1,\dots,\Gamma_n;C)
\]
satisfying appropriate versions of associativity and unit equations.
The equivalence between these two different presentations of multicategories based on either partial or parallel composition operations appears to be folklore.\footnote{A rigorous proof of the equivalence for the operadic case (that is, for one-object multicategories) can be found in a recent monograph by Fresse~\cite[v.~1, ch.~2]{Fresse}.}

The reader may refer to \cite{Lei} for many different examples of multicategories.
Importantly, any monoidal category $(\C,\I,\ot)$ has an underlying multicategory $\M$ with the same objects and with $\M(\Gamma;C) \eqdf \C(\asem{\Gamma},C)$, where $\asem{\Gamma}$ denotes the product of the list of objects $\Gamma = A_1,\dots,A_n$ defined using some bracketing, e.g., $\asem{\Gamma} \eqdf (\ldots (\I \ot A_1) \ldots) \ot A_n$.
%(which bracketing we take does not essentially matter, by the associativity and unit axioms of monoidal categories).
Conversely, a multicategory $\M$ is said to be \defn{representable} just in case for any list of objects $\Gamma = A_1,\dots,A_n$, there is an object $\asem{\Gamma}$ together with a multimap
$m_\Gamma : \multi{\Gamma}{\asem{\Gamma}}$
which is \emph{strong universal} in the sense that there exists a family of bijections
\begin{equation}\label{eq:universal}
\quad\qquad \otLstar_\Gamma : \M(\Delta_0,\Gamma,\Delta_1;C) \overset{\sim}\to \M(\Delta_0,\asem{\Gamma},\Delta_1;C)
\end{equation}
(indexed by $\Delta_0$, $\Delta_1$, and $C$)
whose inverse is the operation
\[
\cut(m_\Gamma,-) : \M(\Delta_0,\asem{\Gamma},\Delta_1;C) \to \M(\Delta_0,\Gamma,\Delta_1;C)
\]
of precomposing with $m_\Gamma$.
%is invertible for all $\Delta_0$, $\Delta_1$, and $C$.
Representable multicategories have been studied closely by Hermida \cite{Hermida2000}, who established among other results a 2-equivalence between the 2-category of monoidal categories and strong monoidal functors and the 2-category of representable multicategories and multifunctors that preserve strong universal multimaps.
Lambek \cite{Lambek1969} already considered essentially the same notion but where $\M$ is supplied with a strong universal nullary map $i : \multi{~}{\I}$ as well as strong universal binary maps $m_{A,B} : \multi{A,B}{A\ot B}$ for every $A$ and $B$.
Lambek called this a \emph{monoidal multicategory}, but to keep the terminology consistent we will refer to it as a \defn{nullary-binary representable} multicategory.
Every representable multicategory is obviously a nullary-binary representable multicategory, but the converse is also true, since binary and nullary strong universal maps can be composed to construct strong universal maps of arbitrary arity.

Before moving on to discuss skewness, as a final remark, let us point out the clear connection (again, already made by Lambek \cite{Lambek1961,Lambek1968,Lambek1969}) between the definition of nullary-binary representable multicategory and the rules of the monoidal sequent calculus (Figure~\ref{fig:lambek}).
In a sense that can be made precise, derivations of the calculus form a free nullary-binary representable multicategory under the $\id$ and $\cut$ rules when considered modulo the appropriate notion of equivalence, where the nullary and binary strong universal maps are derived using the right rules as $\IR$ and $\otR(\id_A, \id_B)$, while the bijections \eqref{eq:universal} correspond directly to the left rules $\IL$ and $\otL$ and the fact that they are invertible rules.
%% \[
%% \infer[\IR]{~ \vdasj \I}{}\qquad  
%% \infer[\otR]{A,B \vdash A \ot B}{\infer[\id]{A\vdash A}{} & \infer[\id]{B \vdash B}{}}
%% \]

In their paper \cite{BL:multi}, Bourke and Lack give a concise definition of ``skew multicategory'' after first introducing a more general notion of \emph{$\mathcal{T}$-multicategory} for any $\Cat$-enriched operad $\mathcal{T}$.
In one formulation, a $\mathcal{T}$-multicategory corresponds to a $\Cat$-enriched multicategory $\M$ equipped with a $\Cat$-enriched multifunctor into $\mathcal{T}$ that is locally a discrete opfibration, in the sense that each functor $\M(A_1,\dots,A_n;C) \to \mathcal{T}_n$ is a discrete opfibration (writing $\mathcal{T}_n$ for the category of $n$-ary operations $\mseq{*,\dots,*}{*}$ of $\mathcal{T}$).
A \emph{skew multicategory} is then just a $\mathcal{R}$-multicategory, where $\mathcal{R}_n$ is defined as equivalent to the arrow category $\mathbf{2}$ for $n > 0$, and the terminal category $\mathbf{1}$ for $n = 0$.
Bourke and Lack also sketch how to unpack this definition into a more conventional description of the structure of a skew multicategory in terms of what they call ``tight'' and ``loose'' multimaps.
In order to illuminate the relationship to the sequent calculus and for its independent interest, we give here a completely explicit but equivalent reformulation of Bourke and Lack's definition.
\begin{definition}
A \defn{skew multicategory} $\M$ consists of:
\begin{itemize}
\item a set $M$ of objects of $\M$
\item for any $S \in M \uplus \{\n\}$, any list of objects $\Gamma = A_1,\dots,A_n \in M$, and object $C \in M$, a set $\M(S\mid\Gamma;C)$ of multimaps; a multimap $f : \flexi{S}{\Gamma}{C}$ is said to be \emph{tight} if $S = A$, and \emph{loose} if $S = \n$
\item for each object $A \in M$, a tight multimap $\id : \tight A~A$
\item for every pair of a (loose or tight) multimap $f : \flexi{S}{\Gamma}{A}$ and a tight multimap $g : \tight{A}{\Delta}{C}$, a multimap $\scut(f,g) : \flexi{S}{\Gamma,\Delta}{C}$; and for every pair of a loose multimap $f : \loose{\Gamma}{A}$ and a multimap $g : \flexi{S}{\Delta_0,A,\Delta_1}{C}$, a multimap $\ccut_{\Delta_0-\Delta_1}(f,g) : \flexi{S}{\Delta_0,\Gamma,\Delta_1}{C}$
\item a family of \emph{comparison functions} $\uf : \M(A\mid\Gamma;C) \to \M(\n\mid A,\Gamma;C)$
\item satisfying all of the equations \eqref{eq:skewmulti1a}--\eqref{eq:skewmulti4c} in Figures~\ref{fig:skeweqns} and \ref{fig:skeweqns-ctd}.
\end{itemize}
\end{definition}
In brief, compared to ordinary multicategories, skew multicategories distinguish ``tight'' multimaps from ``loose'' multimaps, with an inclusion/coercion of the former into the latter, and two kinds of composition.
This also leads to a (perhaps slightly intimidating) proliferation of equations, with Lambek's original four equations \eqref{eq:multi1}--\eqref{eq:multi4} replaced by the eleven equations \eqref{eq:skewmulti1a}--\eqref{eq:skewmulti4c}.
Of course we already encountered these equations in Section~\ref{sec:adequacy}, and as we will see shortly, the skew monoidal sequent calculus indeed defines a (special kind of) skew multicategory.

The main technical difference between our formulation of skew multicategories and Bourke and Lack's is that whereas they postulate a single parallel composition operator acting on multimaps of arbitrary kind -- with some logic for determining whether the result is tight or loose -- we postulate two different partial composition operations $\scut(f,g)$ and $\ccut(f,g)$, respectively for composition \emph{into} the first argument of a tight multimap $g$, and for composition \emph{out of} a loose multimap $f$ into an arbitrary argument of another multimap.
The two formulations are equivalent, however, since Bourke and Lack's  ``polymorphic'' parallel composition operation can be specialized to obtain the partial composition operations $\scut$ and $\ccut$, and conversely, any parallel composition of tight and loose maps can be expressed using an appropriate combination of $\scut$, $\ccut$, and $\uf$.
%(The remaining possibility of composing a tight multimap $f : \tight{A}{\Phi}{B}$ into a multimap $g : \flexi{S}{\Gamma,B,\Delta}{C}$ is covered by first coercing $f$ to a loose multimap and then composing $\ccut(\uf(f),g) : \flexi{S}{\Gamma,A,\Phi,\Delta}{C}$.)

Every skew monoidal category $(\C,\I,\ot)$ gives rise to a skew multicategory with the same objects and with multimaps $\flexi{S}{\Gamma}{C}$ defined as morphisms $\asem{S \mid \Gamma} \to C$, where $S$ is interpreted as an object of $\C$ by taking $\I$ in the loose case $S = \n$, and $\asem{S \mid \Gamma}$ then denotes the left-associated product $(\ldots (S \ot A_1) \ldots) \ot A_n$ for $\Gamma = A_1,\dots,A_n$.
Conversely, a skew multicategory equivalent to one of this form is said to be \emph{left representable}.
\begin{definition}
  A skew multicategory $\M$ is \defn{left representable}
  just in case for any $S \in M \uplus \{\n\}$ and list of objects $\Gamma = A_1,\dots,A_n \in M$, there is an object $\asem{S \mid \Gamma}$ together with a multimap $m_{S,\Gamma} : \flexi{S}{\Gamma}{\asem{S \mid \Gamma}}$ which is \emph{left universal} in the sense that there exist a family of bijections
\begin{equation}\label{eq:left-universal}
  \quad\qquad \otLstar_{S,\Gamma} : \M(S\mid\Gamma,\Delta;C) \overset{\sim}\to \M(\asem{S \mid \Gamma}\mid\Delta;C)
\end{equation}
whose inverse is the operation $\scut(m_{S,\Gamma},-)$ of precomposing with $m_{S,\Gamma}$.
\end{definition}
Analogously to the non-skew case, we also say that a skew multicategory $\M$ is \defn{nullary-binary left representable}
  just in case there is an object $\I$ with a left universal loose multimap $i : \loose{~}{\I}$, as well as an object $A \ot B$ with a left universal tight multimap $m_{A,B} : \tight{A}{B}{A\ot B}$ for every pair of objects $A,B$.
(In the terminology of \cite{BL:multi}, $\M$ is said to ``admit tight binary map classifiers and a nullary map classifier''.)
Since left universal maps are closed under ($\scut$) composition, a skew multicategory is left representable if and only if it is nullary-binary left representable (cf.~Proposition 4.5 of \cite{BL:multi}).
% nz: maybe we do not need to quote B & L's theorem.
The following is stated as one of the main results of \cite{BL:multi}.
\begin{theorem}[Bourke \& Lack \cite{BL:multi}]\label{thm:blequiv}
  There is a 2-equivalence between the 2-category of skew monoidal categories and lax monoidal functors and the 2-category of left representable skew multicategories and skew multifunctors (that do not necessarily preserve left universal multimaps).
\end{theorem}
Note this equivalence also holds replacing lax monoidal functors by strong monoidal functors and simultaneously requiring the skew multifunctors to preserve left universal multimaps.

The skew monoidal sequent calculus and the results of Sections~\ref{sec:seqcalc}--\ref{sec:focusing} can be reunderstood in this categorical language, and vice versa:
\begin{itemize}
\item
  The restriction on the left rules in passing from the standard monoidal sequent calculus to the skew monoidal sequent calculus corresponds precisely to the weakening of the universal property for $\ot$ and $\I$ in passing from representable multicategories to left representable skew multicategories, replacing strong universality by left universality.
\item
  Bourke and Lack illustrate Theorem~\ref{thm:blequiv} by explaining how to construct a left representable skew multicategory from a skew monoidal category \cite[\S6.1]{BL:multi} and conversely \cite[\S6.2]{BL:multi}.
  These constructions follow closely to our proofs of soundness (Theorem~\ref{thm:sound}) and completeness (Theorem~\ref{thm:cmplt}), respectively.
\item
  The implication from nullary-binary left representability to left representability
  is implicitly used in the proof of Lemma~\ref{lem:leftinvstar}.
\item
  The focused sequent calculus gives a direct description of the free nullary-binary left representable skew multicategory over a set of atoms.
  Explicitly, $\mathcal{F}(\Var)$ has formulae as objects and focused derivations $\stseqL{S}{\Gamma}{C}$ as multimaps (notably, one does \emph{not} need to consider equivalence classes of derivations).
  Identity and composition are defined as in the proof of Lemma~\ref{lemma:foc-admits}, and satisfy the skew multicategory equations by Lemma~\ref{lem:skeweqns}.
  The nullary map $i : \stseqL{\n}{~}{\I}$ and the binary maps $m_{A,B} : \stseqL{A}{B}{A \ot B}$ are defined using the admissible (by Lemma~\ref{lemma:foc-admits}) right rules and identity axiom as $i \eqdf \IR$ and $m_{A,B} \eqdf \otR(\id_A,\uf\, \id_B)$, and their left universality is witnessed by the left rules $\IL$ and $\otL$.
\end{itemize}

\section{Conclusion and Future Work}
\label{sec:concl-future}

In this paper, we studied the free skew monoidal category from a
proof-theorist's point-of-view, in the spirit of Lambek's work.
We considered several different deductive systems, ranging from a
categorical calculus directly embodying the definition of the free skew monoidal category,
to a Gentzen-style sequent calculus with two forms of cut rules,
to a cut-free and focused subsystem of canonical derivations.
We learned that although skew monoidal categories have some remarkably subtle properties, the methods of proof theory are surprisingly well-suited for exploring them.
As a consequence of our coherence theorem, the focused sequent calculus
provides a very concrete description of the free skew monoidal category,
suitable for deciding equality of maps and for enumerating the set of maps
between any pair of objects.

We envisage a number of directions for future work.

One obvious direction would be to derive analogous coherence theorems for (non-monoidal) skew closed
categories \cite{Street2013skew} and for skew monoidal closed categories.
This would mean analyzing sequent calculi that correspond to the Lambek calculus with only one implication and without or with conjunction.

%% As mentioned in the Introduction, the analysis we have presented here
%% is closely related to Bourke and Lack's recent characterization of
%% skew monoidal categories as left representable skew multicategories
%% \cite{BL:multi}. Indeed, it appears that the focused sequent calculus
%% gives an explicit construction of the free left representable skew
%% multicategory over a set of atoms. We plan to describe this connection
%% in full detail in another paper.

%NZ: commenting the following out (I guess it's not so important and perhaps distracting to the reader).
% Here we linked the categorical calculus to the focused sequent
% calculus via the unrestricted sequent calculus. That the focused
% calculus is complete with respect to the categorical calculus can also
% be proved directly. The resulting focusing completeness proof can be
% shown to be equal to the post-composition of the unrestricted
% completeness proof with focusing. Whether this leads to a shorter
% development overall remains to be checked.

The fact that there can be multiple maps between a pair of objects in the free skew monoidal category also leads to some interesting questions.
It appears that one can partially order derivations in a canonical way for both the categorical calculus and the sequent calculus.
In particular, we can have a greatest element, i.e., a preferred derivation for any derivable sequent, and have soundness and
completeness preserve these partial orders. Moreover, one may ask whether this
ordering coincides with the canonical ordering induced by Lack and Street's
faithful functor $\Fsk \to \Delta_\bot$, viewing $\Delta_\bot$ as a
2-category with the pointwise ordering on monotone maps.
% We say some words about this plan in Appendix \ref{sec:ineq}.

It is worth mentioning that there are some surprisingly elegant formulae for \emph{counting} different families of maps in the free skew semigroup category (a.k.a. intervals of the Tamari lattice \cite{Chapoton2006}), and so it may be interesting to refine Theorem~\ref{thm:coh:eq} and apply the focused sequent calculus to pursue a quantitative analysis of maps in the free skew monoidal category (similarly to how this was done for Tamari intervals in \cite{Zei}).

Finally, another more speculative direction is to develop sequent
calculi for \emph{higher-dimensional} skew monoidal and/or skew
semigroup categories.
Given the connections between the Tamari order
and the well-studied higher-dimensional polytopes known as
\emph{associahedra} \cite{TamariFestschrift}, it is natural to wonder
whether the methods of proof theory can reveal something new.

\section*{Acknowledgments}
T.U.\ was supported by the Estonian Ministry of Education
and Research institutional research grant no.~IUT33-13. N.V.\ was
supported by a research grant (13156) from Villum Fonden and
the ESF funded Estonian IT Academy research measure
(2014-2020.4.05.19-0001).
N.Z.\ was supported by a Birmingham Fellowship from
the University of Birmingham.
N.V.\ was with the IT University of Copenhagen when the first version 
of this article was written; N.Z.\ was with the University of Birmingham.

\bibliographystyle{natbib}

\begin{thebibliography}{99.}\label{bibliography}

\bibitem{ACU:monnnb} 
Altenkirch, T., Chapman, J.,  Uustalu, T.: 
Monads need not be endofunctors. 
Log.\ Methods Comput.\ Sci., \textbf{11}(1), article 3 (2015). 
\doi{10.2168/lmcs-11(1:3)}

\bibitem{Andreoli92} 
Andreoli, J.-M.: 
Logic programming with focusing proofs in linear logic. 
J. Log. Comput., \textbf{2}(3), 297--347 (1992). 
\doi{10.1093/logcom/2.3.297}

\bibitem{Benabou} 
B\'enabou, J.: 
Cat\'egories avec multiplication.  
C. R. Acad. Sci. Paris, \textbf{256},  1887--1890 (1963) 
Available at \url{http://gallica.bnf.fr/ark:/12148/bpt6k3208j/f1965.image}.

\bibitem{BL:free} 
Bourke, J.,  Lack, S.: 
Free skew monoidal categories. 
J.\ Pure Appl.\ Alg., \textbf{222}, 3255--3281 (2018) 
\doi{10.1016/j.jpaa.2017.12.006}

\bibitem{BL:multi} 
Bourke, J.,  Lack, S.: 
Skew monoidal categories and skew multicategories. 
J.\ Alg., \textbf{506}, 237--266 (2018) 
\doi{10.1016/j.jalgebra.2018.02.039}

\bibitem{BGLS:catss}
Buckley, M., Garner, R., Lack, S., Street, R.:
The {C}atalan simplicial set. 
Math.\ Proc.\ Cambridge Philos.\ Soc., \textbf{158}(12), 211--222 (2014)
\doi{10.1017/s0305004114000498}

\bibitem{Chapoton2006} Chapoton, F.: 
Sur le nombre d'intervalles dans les treillis de {T}amari.  
S\'eminaire Lotharingien de Combinatoire, \textbf{55}, article B55f (2006)
Available at \url{https://www.mat.univie.ac.at/~slc/wpapers/s55chapoton.html} 

 \bibitem{Fresse} 
Fresse, B.: 
Homotopy of Operads and Grothendieck-Teichm\"uller Groups: Parts 1 and 2.
Mathematical Surveys and Monographs, \textbf{217}. Amer.\ Math.\ Soc. (2017)

\bibitem{Gentzen35} 
Gentzen, G.: 
Untersuchungen \"uber das logische Schlie{\ss}en {I}.
Math. Z., \textbf{39}, 176--210 (1935) 
\doi{10.1007/bf01201353} \\
  Translation: Investigations into logical deductions. 
  In: Szabo, M.~E. (ed.), 
  The Collected Papers of Gerhard
  Gentzen, Studies in Logic and the Foundations of Mathematics, \textbf{55}, 
  pp. 68--131. North-Holland (1969)

\bibitem{Girard1991LC} 
Girard, J.-Y.:
  A new constructive logic: classical logic.
  Math.\ Struct.\ in Comput.\ Sci., \textbf{1}(3), 255--296 (1991)
  \doi{10.1017/s0960129500001328}

\bibitem{Hermida2000}
Hermida, C.:
\newblock Representable multicategories.
\newblock Adv.\ Math., \textbf{151}(2), 164--225 (2000)
\doi{10.1006/aima.1999.1877}

\bibitem{Kel:maclcc} 
Kelly, G. M.: 
 On MacLane's conditions for coherence of natural associativities,
  commutativities, etc. 
J.\ Alg., \textbf{1}(4), 397--402 (1964) 
\doi{10.1016/0021-8693(64)90018-3}

\bibitem{LS:skemsw} 
Lack, S., Street, R.: 
Skew monoidales, skew warpings and quantum categories. 
Theor.\ Appl.\ Categ., \textbf{26}, 385--402 (2012) 
Available at 
\url{http://www.tac.mta.ca/tac/volumes/26/15/26-15abs.html}

\bibitem{LS:triosm} 
Lack, S., Street, R.: 
Triangulations, orientals, and skew monoidal categories. 
Adv.\ Math., \textbf{258}, 351--396 (2014) 
\doi{10.1016/j.aim.2014.03.003}

\bibitem{Lam:matss} 
Lambek, J.: 
The mathematics of sentence structure. 
Amer.\ Math.\ Monthly, \textbf{65}(3), 154--170 (1958) 
\doi{10.2307/2310058}

\bibitem{Lambek1961}
Lambek, J.:
On the calculus of syntactic types.
In: Jakobson, R. (ed.), Structure of Language and Its
  Mathematical Aspects, Proc.\ of Symp.\ in Appl.\ Math., \textbf{XII},
  pp.~166--178. Amer.\ Math.\ Soc. (1961)

\bibitem{Lambek1968}
Lambek, J.:
Deductive systems and categories {I}: Syntactic calculus and 
residuated categories.
Math.\ Syst.\ Theory, \textbf{2}(4), 287--318 (1968)
\doi{10.1007/bf01703261}

\bibitem{Lambek1969}
Lambek, J.:
Deductive systems and categories {II}: Standard constructions and
  closed categories.
In: Hilton, P. (ed.), Category Theory, Homology Theory and Their
  Applications, I, Lect.\ Notes in Math., \textbf{86}, 
 pp.~76--122. Springer (1969)
\doi{10.1007/bfb0079385}

\bibitem{Lambek1989} 
Lambek, J.:
Multicategories revisited.  
In: Gray, J. W., Scedrov, A. (eds.)
  Categories in Computer Science and Logic,
  Contemporary Mathematics, \textbf{92}, pp. 217--239.  
Amer. Math. Soc. (1989)

\bibitem{Lei} 
Leinster, T.: 
Higher Operads, Higher Categories.
London Math.\ Soc.\ Lect.\ Note Series, \textbf{298}. 
Cambridge Univ. Press (2004) 
\doi{10.1017/cbo9780511525896} \\
  Preprint version: arXiv preprint 0305049 (2003)
  Available at \url{https://arxiv.org/abs/math/0305049}

\bibitem{ML:natac} 
Mac Lane, S.:
Natural associativity and commutativity. 
Rice Univ.\ Stud., \textbf{49}(4), 28--46 (1963)
  Available at \url{http://hdl.handle.net/1911/62865}.

\bibitem{ML:cwm} 
Mac Lane, S.: 
Categories for the Working Mathematician, 2nd ed.
Graduate Texts in Math., \textbf{5}. Springer (1978) 
 \doi{10.1007/978-1-4757-4721-8}

\bibitem{TamariFestschrift} 
M\"uller-Hoissen, F., Pallo, J.-M., Stasheff, J. (eds.): 
Associahedra, Tamari Lattices and Related Structures: 
Tamari Memorial Festschrift.
  Progress in Mathematics, \textbf{299}.
  Birkh\"auser (2012) 
\doi{10.1007/978-3-0348-0405-9}

\bibitem{Street2013skew} 
Street, R.: Skew-closed categories.
  J. Pure Appl.\ Alg. \textbf{217}(6), pp.~973--988 (2013)
  \doi{10.1016/j.jpaa.2012.09.020}

\bibitem{Szl:skemcb} 
Szlach\'anyi, K.: 
Skew-monoidal categories and bialgebroids. 
Adv.\ Math., \textbf{231}(3--4), 1694--1730 (2012) 
\doi{10.1016/j.aim.2012.06.027}.

\bibitem{Tamari1951phd} 
Tamari, D.:
 Mono\"ides pr\'eordonn\'es et cha\^ines de {M}alcev.
 Th\`ese, Universit\'e de Paris (1951). \\
 Partially published: Bull.\ Soc.\ Math.\ France, 
 \textbf{82}, 53--96 (1954) 
Available at \url{http://eudml.org/doc/86885}

\bibitem{Uus:cohsmc} 
Uustalu, T.: 
Coherence for skew-monoidal categories. 
In: Levy, P. Krishnaswami, N. (eds.)
  Proc.\ of 5th Wksh.\ on Mathematically Structured
    Programming, MSFP 2014, Electron.\ Proc.\ in
    Theor.\ Comput.\ Sci., \textbf{153}, pp. 68--77. Open Publishing
  Assoc. (2014) 
\doi{10.4204/eptcs.153.5}

\bibitem{UVZ:seqsmcMFPS18} 
Uustalu, T., Veltri, N., Zeilberger, N.:
The sequent calculus of skew monoidal categories.
Electron.\ Notes Theor.\ Comput.\ Sci., \textbf{341}, 345--370 (2018)
\doi{10.1016/j.entcs.2018.11.017}

\bibitem{Zei} 
Zeilberger, N.:
A sequent calculus for a semi-associative law. 
In: Miller, D. (ed.) Proc.\ of 2nd Int.\
    Conf.\ on Formal Structures for Computation and Deduction, FSCD
    2017, Leibniz Int.\ Proc.\ in Inform., \textbf{84}, article
  33. Dagstuhl Publishing (2017) 
\doi{10.4230/lipics.fscd.2017.33}

\bibitem{Zei-ext} 
Zeilberger, N.: 
A sequent calculus for a semi-associative law (extended version). 
Log.\ Methods Comput. Sci., \textbf{15}(1), article 9 (2019)
  \doi{10.23638/lmcs-15(1:9)2019}
\end{thebibliography}

\newcommand{\doi}[1]{\href{http://dx.doi.org/#1}{doi: #1}}

\end{document}